\documentclass[11pt,letterpaper]{article}

\usepackage[dvipsnames]{xcolor}
\usepackage{mathtools}
\usepackage{xspace}
\usepackage{booktabs}
\usepackage{authblk}
\usepackage{parskip}

\usepackage[square, semicolon]{natbib}
\usepackage{geometry}
\usepackage[colorlinks = true, linkcolor=NavyBlue,citecolor=ForestGreen]{hyperref}

\usepackage[tt=false]{libertine}

\usepackage{nicefrac}       % compact symbols for 1/2, etc.
\usepackage{microtype}      % microtypography

\usepackage{algorithm}
\usepackage{algpseudocodex} % 

\usepackage{amsmath}
\usepackage{amsthm}
\usepackage{amsfonts}
\usepackage{amssymb}

\usepackage{enumitem}
\usepackage{bm}
\usepackage{setspace}

\theoremstyle{plain}
\newtheorem{theorem}{Theorem}[section]
\newtheorem{lemma}[theorem]{Lemma}
\newtheorem{corollary}[theorem]{Corollary}
\newtheorem{definition}[theorem]{Definition}

\theoremstyle{remark}

\newtheoremstyle{nonitalic}{3pt}{3pt}{}{0pt}{\bfseries}{.}{ }{}
\theoremstyle{nonitalic}

\newcommand{\aaa}{\mathcal{A}}

\newcommand{\perfect}{\text{left-perfect matching}\xspace}

\title{Online Fair Division Meets Reordering Buffers\thanks{A conference version of this work appears in the Proceedings of the 40th Conference on Neural Information Processing Systems (NeurIPS 2026).}}

\author[1,2]{Georgios Amanatidis}
\author[3]{Giulio Giaconi}
\author[1,2,4]{Evangelos Markakis}
\author[2]{Nicos Protopapas}

\affil[1]{Department of Informatics, Athens University of Economics and Business, Athens, Greece}
\affil[2]{Archimedes, Athena Research Center, Athens, Greece}
\affil[3]{Quantum and Advanced Compute Group, HSBC Holdings Plc., United Kingdom}
\affil[4]{Input Output, Athens, Greece}

\hypersetup{
pdftitle={Online Fair Division Meets Reordering Buffers},
pdfsubject={},
pdfauthor={Georgios Amanatidis, Giulio Giaconi, Evangelos Markakis, Nicos Protopapas},
pdfkeywords={Online Fair Division, Discrete Fair Division, Mixed Manna, Temporal Envy-Freeness.},
}

\newcommand{\efo}{\ensuremath{\textrm{\MakeUppercase{ef}1}}\xspace}
\newcommand{\ef}{\ensuremath{\textrm{\MakeUppercase{ef}}}\xspace}
\usepackage{tikz}
\newcommand*\circled[1]{\tikz[baseline=(char.base)]{
		\node[shape=circle,draw,inner sep=1pt] (char) {#1};}}
\newcommand*\scircled[1]{\tikz[baseline=(char.base)]{
		\node[shape=circle,draw,inner sep=2.5pt] (char) {#1};}}

\begin{document}
\maketitle

\begin{abstract}
We study the \textit{online fair division} of indivisible \textit{mixed manna} among agents with additive valuation functions. Under the standard online model, at each time step an indivisible item arrives; each agent may assign it a positive, negative, or zero value, and it must be irrevocably allocated, before the arrival of the next item. At the same time, we also wish to maintain some fairness guarantee, and in this work we focus on \textit{envy-freeness} (EF) and one of its most prominent relaxations, \textit{envy-freeness up to one item} (EF1). Given the strong negative and the scarce positive results for this problem without additional assumptions, we augment our algorithms with \textit{buffers} that can store and rearrange a limited number of items. This setting interpolates naturally between the fully online case (no buffer) and the fully offline case (a buffer large enough to hold all items). We show that algorithms equipped with reasonably sized buffers can achieve strong guarantees for personalized $k$-value instances, i.e., instances in which each agent assigns at most $k$ distinct values to items. In particular, we construct allocations that are EF1 at every time step and EF at most time steps, using a buffer of size linear in $k$ and in the number of agents. Our approach relies on novel combinatorial arguments and on constructing a sequence of envy-free matchings that allocates most items. Finally, we extend our results to general additive valuation functions, with a dependence on the largest per-agent ratio between two values of the same sign, and we also identify limitations of our approach via impossibility results on the use of buffers with smaller size.  
\end{abstract}

\section{Introduction}

Our work concerns the fair allocation of indivisible items to a set of interested agents. Fair division has attracted significant interest within the broader algorithmic game theory community, with a sizeable volume of recent literature, as can also be seen by surveys such as \cite{AmanatidisABFLMVW23}, \citet{liu2024mixed}, \cite{10.1145/3505156.3505162} and \cite{Biswas2023}.
The emergence of additional motivating applications, including among others 
food donation programs \citep{Mertzanidis0V24}, 
further contributes to the growing momentum of the relevant community.
This has naturally led to a variety of fair division models, dependent on the type of items to be allocated, the type of preferences, but also on possible constraints on the allocation space and the fairness notions of interest.

In this work, we consider an online scenario where the items are not available from the beginning but instead arrive sequentially, one by one. This can be seen as a more realistic model, compared to the more commonly studied offline model, and is motivated by scheduling applications and other problems where resources are released over time. Therefore, an algorithm under this model needs to maintain a partial allocation that is being updated as time progresses, until there are no further arrivals.  
Furthermore, regarding the type of items, we focus on the most general setting  that is commonly referred to as \emph{mixed manna}, where an item can be  valued either non-negatively (perceived as a \textit{good}) or non-positively (perceived as a \textit{chore}) by an agent.  
Finally, our target fairness notions are \textit{envy-freeness} (EF) and one of its most prominent relaxations, \textit{envy-freeness up to one item} (EF1). Given these considerations, ideally we would  like to have algorithms that maintain temporal fairness, i.e., the allocations they produce are EF or EF1 \textit{in every time step} during their execution.

If we follow the classic model of online algorithms, every time a new item arrives, it must be \emph{irrevocably} allocated to some agent. Unfortunately, such a constraint already makes the problem very challenging. In particular, there are strong impossibility results showing that one cannot hope for any reasonable approximation guarantees on EF1 \citep[see, e.g.,][]{HePPZ19,wang2025online}. There are only scarce positive results for very special cases, and this highlights the limitations that online algorithms are facing for this problem without additional assumptions.

One way to circumvent these impossibilities is by augmenting an algorithm with additional power, which can come in various forms. As a first example, we could assume some limited form of lookahead access, i.e., the algorithm can see the values of the items that will come (say for a small number of future steps), but still needs to take an irrevocable decision on how to allocate the item that arrived in the current step. This can be meaningful especially in applications where we can estimate future values (e.g., via some predictions). As another example, an algorithm may be equipped with a buffer, that can store items, as introduced in the context of online job scheduling \citep{kellerer1997semi}, and referred to as a \emph{reordering buffer}. In this case, there is no need to allocate the currently arriving item right away, but instead we could store items and decide later on how to allocate them. This can be motivated partly by scheduling problems, but also by any other application where one may need to temporarily store resources, so as to produce a better allocation.

\subsection{Our Contribution}

Our work initiates the study 
of using reordering buffers, as per \cite{kellerer1997semi}, in online fair division. We note that the model of buffers interpolates naturally between the fully online case (no buffer) and the fully offline case (a buffer large enough to hold all items). Our main takeaway message is that the use of appropriately sized buffers can greatly help bypass the strong negative results of traditional online algorithms. Our main results can be summarized as follows.

\begin{enumerate}[topsep=1pt,leftmargin=18pt,itemsep=0pt]
    \item {\bf Impossibility results for the fully online case and for algorithms with lookahead.} We begin our exposition in Section \ref{sec:impossibility_lookahead}, where we demonstrate that without the use of buffers, there are severe impossibility results. For the fully online case, this is already known even for 3-value instances. We prove that even with a lookahead  almost as large as the total number of items, we cannot have algorithms that produce temporal-EF1 allocations. Furthermore, even asking for approximate EF1 guarantees is not possible with limited lookahead.  
    \item {\bf Positive results with the use of a buffer.} In Section \ref{sec:main}, we obtain our main positive results. We consider two allocation models, based on whether the algorithm can allocate items in batches or one by one. In both models, we show that for $n$ agents with $k$-valued preferences, a buffer of size $(n-1)k$, which is independent of the total number of items, suffices to obtain a temporal-EF1 allocation. Our Algorithm \ref{alg:online-buffer} also guarantees that the allocation is EF at least every $n$ steps (except possibly for the last $(n-1)k$ steps in the model where items are allocated one by one). Our approach relies on novel combinatorial arguments, leveraging Hall's Theorem for constructing a sequence of envy-free matchings that allocates most items. 
    In Section \ref{sec:general_additive}, we extend our results to general additive valuations for goods or for chores, at the expense of a multiplicative loss, dependent on the largest ratio between two non-zero values of the same agent.
    \item {\bf Lower bounds on the necessary buffer size.} In Section \ref{sub:impossibility_buffer}, we investigate whether the size of the buffer in our positive results can be improved. In the model where items are allocated in batches we show that our result is tight, whereas in the model where items are allocated one by one we show that a dependence on $\sqrt{k}$ is necessary.
\end{enumerate}

\subsection{Further related work}

\noindent \textbf{Mixed manna.} Our main positive results work for the case where each agent might have positive, negative or zero value for an item. This is coined as \emph{mixed manna} by \citet{BogomolnaiaMSY17} and it is known that EF1 allocations always exist in the offline setting \citep{aziz2022fair}. Prior to our work, little was known about mixed manna in the online setting. A notable exception is due to~\citet{ElkindLLNT25}, where it is shown that with full lookahead temporal-EF1 allocations exist for two agents.

\noindent \textbf{Online fair division.} 
Online fair division has received increasing attention in recent years, despite important impossibility results. The works closest to ours, one way or another, are those of \citet{HePPZ19}, \citet{CooksonES25}, \citet{ElkindLLNT25}, and \citet{choi2026temporal}.

\citet{HePPZ19} study a model with reallocations, where the objective is to keep their number small. They show that it is impossible to maintain temporal-\efo without  a large number of reallocations. This model, however, does not capture the role of a buffer.
The works of \citet{ElkindLLNT25} and \citet{CooksonES25} formalize the notion of \emph{temporal fairness} and provide results for various special cases. We note that all three previous papers utilize full lookahead, at least for some of their results. 
The recent paper of \citet{choi2026temporal} examines various fairness criteria in their temporal form. More closely related to our work, they allow all items to be delayed up to a time bound, to get improved approximation guarantees---in contrast, we allow unbounded delays, but only for a bounded number of items at any time (in our positive results, at most $(n-1)k$ items, independently of the total number of items). A similar idea was also used in~\cite{wang2025online} for personalized $2$-value instances.

Other works use distributional assumptions, randomization, online learning, or additional information such as reliable advice or unreliable predictions~\citep{aleksandrov2020onlinesurvey, benade2024fair,YamadaKAI24,ProcacciaS024,schiffer2025improved,neoh2025online,choo2025approximate,melissourgos2025online}. A related but technically different line considers divisible online items~\citep{GkatzelisPT21,Barman0M22,BanerjeeGGJ22,BanerjeeGHJM023}; the connection is limited, since divisibility makes positive results substantially easier. Other directions study objectives such as maximin share, Nash welfare, and envy-freeness with subsidies~\citep{ZhouBW23,SongTWZ25,kulkarni2025online,wang2025online}, or settings where agents rather than items arrive online~\citep{KalinowskiNW13,KashPS14,ijcai2019p773,SinclairJBY23,VardiPF22,BanerjeeHS24,KulkarniMS25}.

\noindent \textbf{Use of a buffer in online algorithms.} 
Buffers have long been used to mitigate the limitations of online decision-making, with applications in scheduling~\citep{dwibedy2022semi,englert2008power,epstein2011max,kellerer1997semi}, web caching~\citep{albers2004new,feder2004combining}, and bin packing~\citep{zheng2015nf}.

%%%%%%%%%%%%%%%%%%%%%%%%%%%%%%%%%%%
\section{Preliminaries}
\label{sec:prelims}
%%%%%%%%%%%%%%%%%%%%%%%%%%%%%%%%%%%

For any $z \in \mathbb{N}_{> 0}$ we use $[z]$ to denote the set $\{1,2,\dots,z\}$.
We consider a set of $n$ agents, $N = [n]$, and a set $M=[m]$ of $m$ \textit{indivisible} items for some $n \in \mathbb{N}_{\geq 2}$ and $m  \in \mathbb{N}$.
A (partial) allocation $\aaa$ in our model is any ordered partition of (a subset of) the items into $n$ subsets, $\aaa = (A_1, \dots, A_n)$, where $A_i$ is the \textit{bundle} of agent $i$. 

We consider agents with additive valuation functions, i.e., each agent $i \in [n]$ associates a value $v_i(\{g\}) \in \mathbb{R}$ for each item $g \in [m]$, and for any given set $A\subseteq M$, $v_i(A)=\sum_{g \in A} v_i(\{g\})$; from this point onward, we will write $v_i(g)$ instead of $v_i(\{g\})$, for the sake of readability.
When the valuation functions take both positive and negative values, we refer to the items of $M$ as \textit{mixed manna}. We  also consider the monotone special cases of having only \emph{goods}, where $v_i(g)\geq 0$ for all $i \in N$ and all $g \in M$, and the respective case of \emph{chores} where $v_i(g) \leq 0$ for all $i \in N$ and all $g \in M$.

A specific restriction with respect to the valuations, which is central in this work, is the following.

\begin{definition}[Personalized $k$-value instances]
We say that an instance of the problem is a \emph{personalized $k$-value instance} if for any $i \in N$, the valuation function $v_i$ is additive and there exist real numbers
$\alpha_{i_1} \ge \alpha_{i_2} \ge \dots \ge \alpha_{i_k}$, such that for any $g \in M$, it holds that
$v_i(g) \in \{\alpha_{i_1}, \alpha_{i_2}, \dots, \alpha_{i_k}\}$. When $\alpha_{i_\ell} = \alpha_\ell$ for all $i \in N$ and all $\ell \in [k]$, we call this a \emph{$k$-value instance}.
\end{definition}

In a nutshell, personalized $k$-value instances cover situations where each agent has at most $k$ distinct valuation levels in their valuation function and these valuation levels may be different per agent.

An ideal solution concerning fairness is that no agent prefers another agent's bundle to their own. 

\begin{definition}[Envy-freeness (EF)]
An allocation $\aaa = (A_1, \dots, A_n)$ is \emph{envy-free (EF)} if for every pair of agents $i,j\in N$, it holds that $v_i(A_i) \ge v_i(A_j)$.
\end{definition}

It is well-known that envy-free allocations do not always exist. 
Therefore, several relaxations have been considered as alternative solutions. Among these, the one we focus on in our work is the well-known criterion of EF1, and in particular, its \emph{temporal} form (see Definition \ref{def:temporal}).

\begin{definition}[Envy-freeness up to one item (EF1)]
An allocation $\aaa = (A_1, \dots, A_n)$ is \emph{envy-free up to one item (EF1)} if for every pair of agents $i, j \in N$, either $i$ does not envy $j$, or there exists an item $g \in A_i \cup A_j$ such that $v_i(A_i \setminus \{g\}) \ge  v_i(A_j \setminus \{g\})$.
\end{definition}

The notion of \efo was initially introduced by \citet{LMMS04} and formalized by \citet{Budish11} for goods, and was later generalized by \citet{aziz2022fair} for mixed manna. It captures the fact that any envy that some agent $i$ may have towards another agent $j$ can be eliminated by removing either a positively valued item from $A_j$ or a negatively valued item from $A_i$. 
Note that this definition collapses to the standard definitions for EF1 for the goods-only case (see e.g.,~\cite{AmanatidisABFLMVW23}) or the chores-only case (see e.g.,~\cite{guo2023survey}). 
In the presence of goods-only or chores-only instances,
we also consider the natural approximate versions of EF1. 

\begin{definition}[$\beta$-EF1 for goods / chores]\label{def:rho-EF1}
Let $\beta \in (0,1]$. In a goods-only (resp.~chores-only) instance, an allocation $\aaa = (A_1, \dots, A_n)$ is $\beta$-EF1 if for every pair of agents $i, j \in [n]$, either $i$ does not envy $j$, or there exists an item $g \in A_j$ (resp.~$g \in A_i$) such that $v_i(A_i ) \ge \beta \cdot v_i(A_j \setminus \{g\})$ (resp.~$\beta \cdot v_i(A_i \setminus \{g\}) \ge v_i(A_j )$). Similarly, such an allocation is \emph{$\beta$-EF} if $v_i(A_i ) \ge \beta \cdot v_i(A_j)$ (resp.~$\beta \cdot v_i(A_i) \ge v_i(A_j )$) for every pair of agents $i, j \in [n]$.
\end{definition}

%%%%%%%%%%%%%%%%%%%%%%%%%%%%%%%%%%%%%%%%%%%%%%%
\subsection{Online Fair Division}
%%%%%%%%%%%%%%%%%%%%%%%%%%%%%%%%%%%%%%%%%%%%%%%
The most common setting in fair division is \emph{offline}: the whole set $M$ of items is known and available to be allocated immediately. We consider an \emph{online} environment,
where the set of agents is static but the items arrive sequentially: in each time step $t=1,2,\ldots$ the item $g_t$ arrives, and we need to \textit{irrevocably} allocate it to one of the agents, usually \textit{immediately}. The value each agent has for the item becomes known only upon its arrival. 
Given the online nature of the problem, we no longer care primarily for the fairness guarantees of the final, complete allocation, but for the corresponding guarantees \textit{in every time step}, that is, for \emph{temporal fairness} as it was introduced by \citet{ElkindLLNT25,CooksonES25}.

\begin{definition}\label{def:temporal}
Consider a sequence of partial allocations $\mathcal{A}^t = (A_1^t, A_2^t, \dots, A_n^t)$, for $t \in \mathbb{Z}_{\ge 0}$, such that $A_i^t \subseteq A_i^{t+1}$ for any $i \in N$ and any $t \ge 0$.
If $\mathcal{A}^t$ is $\beta$-EF1 for all $t \in \mathbb{Z}_{\ge 0}$, then we say that the sequence of allocations $(\mathcal{A}^t)_{t \ge 0}$ is $\beta$-temporal-EF1.
\end{definition}

We use the simplest notation of temporal-EF1 when $\beta=1$. In some cases, we may have a stronger fairness guarantee for all but the very last time step.

\begin{definition}\label{def:temporal-ef}
	Consider a sequence of partial allocations $\mathcal{A}^t = (A_1^t, A_2^t, \dots, A_n^t)$, for $t \in \{0,1,\ldots,\tau\}$, such that $A_i^t \subseteq A_i^{t+1}$ for any $i \in N$ and any $t \le \tau-1$.
	If $\mathcal{A}^t$ is $\beta$-EF for all $t \le \tau-1$ and $\mathcal{A}^{\tau}$ is a complete $\beta$-EF1 allocation, then we say that the sequence is $\beta$-temporal-\ef/\efo.
\end{definition}

We make no distributional assumptions about the arrival of the items, and follow a worst-case analysis. In some of our results, in Section \ref{sec:impossibility_lookahead}, the algorithm can view the values of some items ahead of their arrival time. We say that an online algorithm is augmented with a \emph{lookahead} of size $\ell$ if at time step $t$, the valuations for the items $g_t,\ldots,g_{t+\ell}$ are revealed, while the algorithm can still allocate only the item $g_t$. 

The main enhancement we assume for our online algorithms is the use of \textit{reordering buffers}. 
A buffer of size $b$ is essentially a set $B$, where we are allowed to temporarily store up to $b$ items to facilitate an online algorithm. The items can be stored for as many time steps as needed but everything must be allocated eventually. Crucially, if the buffer is full, i.e., $|B|=b$, then the algorithm cannot store a new item, unless it immediately allocates at least one item from $B$.
An interesting implication of using buffers is that now it is not always necessary to immediately allocate an item. Moreover, it is  possible to allow the algorithm to allocate multiple items at once. With respect to this, we consider two modes of allocation from the pool of available items (i.e., items in the buffer and the newly arrived item):
\begin{itemize}[topsep=-2pt,leftmargin=16pt,itemsep=0pt]
    \item[-] \textbf{Sustained Allocation of Singletons} (SAS): At most one available item can be allocated. This is closer to the practice in the literature of online fair division (without a buffer) and is illustrative of  the challenges of achieving temporal fairness without allocating bundles of items.
    \item[-] \textbf{Deferred Allocation of Batches} (DAB): Any subset of the available items can be allocated at any time. Giving the items in appropriately selected batches turns out to be powerful enough to allow us to obtain very strong fairness guarantees that are not typical in online fair division. 
\end{itemize}
It is easy to see that our model interpolates between standard online fair division, where $b=0$, and offline fair division, where $b\ge m$. Interestingly, for $b\ge m/2$ the problem of obtaining temporal-\efo allocations is relatively easy by essentially reducing the problem to its offline counterpart \citep{ElkindLLNT25} but, below that threshold, utilizing the buffer seems to be completely nontrivial.

%%%%%%%%%%%%%%%%%%%%%%%%%%%%%%%%%%%%%
\section{Impossibility Results Leading to the Use of Buffers}
\label{sec:impossibility_lookahead}
%%%%%%%%%%%%%%%%%%%%%%%%%%%%%%%%%%%%%

In this section, we provide justification on why one needs to go beyond the standard model of online algorithms in order to have fairness guarantees in online fair division. We have already mentioned in the introduction that there are strong impossibility results not only for general additive valuation functions \citep{HePPZ19} but for $k$-value instances as well, even for $k = 3$ \citep{wang2025online}. We restate a parametric version of the latter result, as we are going to refer to it later in Section \ref{sec:general_additive}.

\begin{theorem}[Follows from \citet{wang2025online}]\label{thm:3-value_impossibility}
Let $c>1$ and $\varepsilon>0$. No deterministic online algorithm can always compute $(1/\sqrt{c} + \varepsilon)$-temporal-\efo allocations for $3$-value instances with values $1, \sqrt{c}, c$, or $-c, -\sqrt{c}, -1$, even when $n=2$.
\end{theorem}

For the proof of Theorem \ref{thm:3-value_impossibility} one needs fairly simple instances that exploit an algorithm's lack of knowledge of the future. One possible remedy for this, which was recently introduced in online fair division, is the (partial) knowledge of the future \citep[see, e.g.,][]{HePPZ19,ElkindLLNT25,amanatidis2025online}.
Here we first show that  in order to have any hope to achieve a nontrivial  guarantee, one needs a lookahead of \textit{nearly} $k$ future items, where $k$ is the number of distinct value levels in the instance.

\begin{theorem}\label{thm:k-value_impossibility_lookahead}
Let $c\ge 2$ and $\varepsilon>0$. No deterministic online algorithm with a lookahead $\ell \le k-3$ can always compute $(1/c  +\varepsilon)$-temporal-\efo allocations for $k$-value instances, even when $n=2$ and all items are only goods or only chores.
\end{theorem}

\begin{proof}
    We are going to show in full detail the case where all items are goods; the case where everything is a chore is very similar and we are only going to highlight the differences.

Suppose we have a deterministic allocation algorithm $\mathcal{A}$ with a lookahead of size $\ell \le k-3$ that always computes a $(1/c  +\varepsilon)$-temporal-\efo allocation when given a $k$-value instance with only goods.
We are going to construct a sequence of items that forces $\mathcal{A}$ to fail to produce such an allocation within at most $\ell + 4$ time steps. Moreover, this sequence will use at most $k$ distinct values, leading to a contradiction.

Consider first the goods $g_1, g_2, \ldots, g_{\ell+1}$, such that 
$v_1(g_i) = v_2(g_i) = c^{i-1}$ for $i\in [\ell+1]$, as shown below; the vertical line indicates the end of the initial view of algorithm $\mathcal{A}$ (initial item, $g_1$, and $\ell$ additional items). 
% \begin{table}[h!]
	\begin{center}  
	\begin{tabular}{lcccccc|c}
		& $g_1$ & $g_2$ & $g_3$   &\ldots & $g_{\ell}$ & $g_{\ell+1}$  & \ldots \\[5pt]
		agent 1: \hspace{1em} & $1$ & $c$ & $c^2$ &  \ldots & $c^{\ell-1}$ & $c^{\ell}$ & \ldots \\
		agent 2: \hspace{1em} & $1$ & $c$ & $c^2$ & \ldots & $c^{\ell-1}$ & $c^{\ell}$ & \ldots 
	\end{tabular}
	\end{center}
% \end{table}
Due to the symmetry of the visible part of the instance, it is without loss of generality to assume that good $g_1$ gets allocated to agent 1. We claim that once this happens, the algorithm must alternate between the two agents, thus giving all the odd-indexed goods up to (and including) $g_{\ell+1}$ to agent 1 and all the even-indexed ones to agent 2. 
Indeed, suppose this is not the case. Then there are two consecutive goods given to the same agent; let $g_i, g_{i+1}$ be the first goods for which this happens. If $i$ is odd, then $g_i, g_{i+1}$ are both given to agent 1 and we have 
\begin{align*}
v_2(A_1^{i+1}) &= v_2(\{g_1, g_3, \ldots, g_i, g_{i+1}\}) = \sum_{j=0}^{(i-1)/2} \!\!c^{2j}  +c^i = 1+ c\cdot\!\! \!\sum_{j=1}^{(i-1)/2} \!\!c^{2j-1}  +c^i \\ 
&= c\cdot v_2(\{g_2, g_4, \ldots, g_{i-1}\})  +c^i  +1 = c\cdot v_2(A_2^{i+1})  +c^i +1 \,.
\end{align*}
That is, the allocation $(A_1^{i+1},A_2^{i+1})$ is not even $1/c$-\efo, contradicting the choice of $\mathcal{A}$. 
If $i$ is even, then $g_i, g_{i+1}$ are both given to agent 2 and we have 
\begin{align*}
v_1(A_2^{i+1}) &= v_1(\{g_2, g_4, \ldots, g_i, g_{i+1}\}) = \sum_{j=1}^{i/2} c^{2j-1}  +c^i = c\cdot\!\!\! \sum_{j=0}^{(i-2)/2} \!\!c^{2j}  +c^i \\ 
&= c\cdot v_1(\{g_1, g_3, \ldots, g_{i-1}\})  +c^i   = c\cdot v_1(A_1^{i+1})  +c^i \,.
\end{align*}
That is, the allocation $(A_1^{i+1},A_2^{i+1})$ is only $1/c$-\efo, again contradicting the choice of $\mathcal{A}$. 
We conclude that algorithm $\mathcal{A}$ allocates the first $\ell+1$ goods so that $A_1^{\ell+1} = \{g_1, g_3, \ldots\}$ and $A_2^{\ell+1} = \{g_2, g_4, \ldots\}$. The next 3 goods depend on the parity of $\ell$:
% \begin{table}[h!]
\begin{center}  
\begin{tabular}{lccccccccccc}
	                    &  & $g_{\ell+1}$ & $g_{\ell+2}$   & $g_{\ell+3}$ & $g_{\ell+4}$ &   &  & $g_{\ell+1}$ & $g_{\ell+2}$   & $g_{\ell+3}$ & $g_{\ell+4}$ \\[5pt]
agent 1:  &\ldots &  $c^{\ell}$ & $c^{\ell+1}$ &  $c^{\ell+4}$ & $c^{\ell+4}$ & \ \ or\ \  & \ldots  & \scircled{$c^{\ell}$} &$c^{\ell+4}$ & $c^{\ell+1}$ & $c^{\ell+4}$ \\
agent 2:  &\ldots &  \scircled{$c^{\ell}$} &$c^{\ell+4}$ & $c^{\ell+1}$ & $c^{\ell+4}$ &  & \ldots & $c^{\ell}$ & $c^{\ell+1}$ &  $c^{\ell+4}$ & $c^{\ell+4}$ 
\end{tabular}
\end{center}
% \end{table}
We are going to analyze the case on the left, where $\ell$ is odd. The case where $\ell$ is even on the right, although not exactly symmetric, is completely analogous.  

Using the exact same calculations as above, we get that if $g_{\ell+2}$ was given to agent 2 we would have $v_1(A_2^{\ell+2}) = c\cdot v_1(A_1^{\ell+2})  + v_1(g_{\ell+2})$, i.e., the allocation $(A_1^{\ell+2},A_2^{\ell+2})$ would only be $1/c$-\efo; so $g_{\ell+2}$ is given to agent 1. Similarly, if $g_{\ell+3}$ was given to agent 1 we would have $v_2(A_1^{\ell+3}) = c\cdot v_2(A_2^{\ell+3})  + v_2(g_{\ell+2}) + 1$, i.e., the allocation $(A_1^{\ell+3},A_2^{\ell+3})$ would not even be $1/c$-\efo; so $g_{\ell+3}$ is given to agent 2. 
Now, whoever gets $g_{\ell+4}$, the resulting allocation is at most $1/c$-\efo. To see this, right before $g_{\ell+4}$ is allocated, we have 
\[v_1(A_1^{\ell+3}) = \sum_{j=0}^{(\ell+1)/2} \!\!c^{2j} = \frac{c^{\ell+3}-1}{c^2 - 1} \le \frac{c^{\ell+3}}{3} < \frac{1}{c}\, v_1(g_{\ell+4}) \le \frac{1}{c}\, v_1(A_2^{\ell+3}) \,,\]
and
\[v_2(A_2^{\ell+3}) = \sum_{j=1}^{(\ell+1)/2} \!\!c^{2j-1} + c^{\ell+1} = c\,\frac{c^{\ell+1}-1}{c^2 - 1} + c^{\ell+1}  \le \frac{2c^{\ell+3}}{3} < \frac{1}{c}\, v_2(g_{\ell+4}) \le \frac{1}{c}\, v_2(A_1^{\ell+3}) \,, \]
where the first inequality in each case follows from the fact that $c\ge 2$ and from simple calculations. In any case, by the time $g_{\ell+4}$ is given, the allocation fails to be $(1/c  +\varepsilon)$-\efo.

In the case of chores the construction is very similar, starting with 
% \begin{table}[h!]
	\begin{center}  
	\begin{tabular}{lcccccc}
		& $g_1$ & $g_2$ & $g_3$   &\ldots & $g_{\ell}$ & $g_{\ell+1}$  \\[5pt]
		agent 1: \hspace{1em} & $-1$ & $-c$ & $-c^2$ &  \ldots & $-c^{\ell-1}$ & $-c^{\ell}$  \\
		agent 2: \hspace{1em} & $-1$ & $-c$ & $-c^2$ &  \ldots & $-c^{\ell-1}$ & $-c^{\ell}$ 
	\end{tabular}
	\end{center}
% \end{table}
and arguing as before we show that algorithm $\mathcal{A}$ gives
all the odd-indexed chores up to (and including) $g_{\ell+1}$ to agent 1 and all the even-indexed ones to agent 2. (We now look at $v_1(A_1^{i+1})$ and $v_2(A_2^{i+1})$ instead of $v_2(A_1^{i+1})$ and $v_1(A_2^{i+1})$ but the calculations are essentially the same.)
The main difference is that the values of the additional chores are switched in the following sense:
% \begin{table}[h!]
\begin{center}  
\begin{tabular}{lccccccccccc}
	             &       & $g_{\ell+1}$ & $g_{\ell+2}$   & $g_{\ell+3}$ & $g_{\ell+4}$ &  &    & $g_{\ell+1}$ & $g_{\ell+2}$   & $g_{\ell+3}$ & $g_{\ell+4}$ \\[5pt]
agent 1: &$\cdots$   &  $-c^{\ell}$ & $-c^{\ell+4}$ &  $-c^{\ell+1}$ & $-c^{\ell+4}$ & \ \ or\ \  & $\cdots$   & \circled{$-c^{\ell}$} &$-c^{\ell+1}$ & $-c^{\ell+4}$ & $-c^{\ell+4}$ \\
agent 2:  & $\cdots$ &  \circled{$-c^{\ell}$} &$-c^{\ell+1}$ & $-c^{\ell+4}$ & $-c^{\ell+4}$ &  & $\cdots$   & $-c^{\ell}$ & $-c^{\ell+4}$ &  $-c^{\ell+1}$ & $-c^{\ell+4}$ 
\end{tabular}
\end{center}
% \end{table}
Still, like before, in the case on the left, where $\ell$ is odd (the case where $\ell$ is even being again completely analogous), $g_{\ell+2}$ is given to agent 1,  $g_{\ell+3}$ is given to agent 2 and no matter who gets $g_{\ell+4}$ the final allocation is at most $1/c$-\efo. 
To see the latter, note that
\[v_1(A_2^{\ell+3}) = \sum_{j=1}^{(\ell+1)/2} \!\!-c^{2j-1} - c^{\ell+1} = -c\,\frac{c^{\ell+1}-1}{c^2 - 1} - c^{\ell+1}  \ge -\frac{2c^{\ell+3}}{3} > \frac{1}{c}\, v_1(g_{\ell+4}) \ge \frac{1}{c}\, v_1(A_1^{\ell+3}) \,, \]
and
\[v_2(A_1^{\ell+3}) = \sum_{j=0}^{(\ell+1)/2} \!\!-c^{2j} = -\frac{c^{\ell+3}-1}{c^2 - 1} \ge -\frac{c^{\ell+3}}{3} >  \frac{1}{c}\, v_2(g_{\ell+4}) \ge \frac{1}{c}\, v_2(A_2^{\ell+3}) \,,\]
so, by the time $g_{\ell+4}$ is given, the allocation fails to be $(1/c  +\varepsilon)$-\efo.
\end{proof}

One might assume that the impossibility stems from the fact that the horizon of the instance is comparable to $\ell$ and/or $k$. However, when one primarily cares for exact temporal-\efo allocations, as is the case here, impossibility results persist even when $m\gg k$ and most of the future information is known up front.
Theorem \ref{thm:3-value_impossibility_lookahead} below is somewhat surprising, given that when the whole sequence can be seen from the beginning (i.e., when $\ell = m - 1$), it is known that a temporal-\efo allocation can  always be computed for two agents, even when they have general additive valuation functions \citep{HePPZ19}.

\begin{theorem}\label{thm:3-value_impossibility_lookahead}
No deterministic online algorithm with a lookahead $\ell \le m-4$ can always compute a temporal-\efo allocation for $3$-value instances, even when $n=2$ and all items are only goods or only chores.
\end{theorem}

\begin{proof}%[\textbf{Proof of Theorem \ref{thm:3-value_impossibility_lookahead}}]
Like in the proof of Theorem \ref{thm:k-value_impossibility_lookahead}, we present in full detail the case where all items are goods; the case of chores is very similar and we will highlight the differences at the end of the proof.

Suppose we have a deterministic allocation algorithm $\mathcal{A}$ with a lookahead of size $\ell = m-4$ that always computes a temporal-\efo allocation when given a $3$-value instance with goods. (Note that the case where $\ell < m-4$ is covered, in the sense that an algorithm can always simulate a smaller lookahead  by just ignoring some of the future values it sees.)
We are going to construct a sequence of items that eventually forces $\mathcal{A}$ to fail to produce such an allocation. Moreover, this sequence will use at most $3$ distinct values, $1, c, c^2$, for $c\ge3$, leading to a contradiction.

Assume first that $m$ is even. Consider  the goods $g_1, g_2, \ldots, g_{m-3}$, such that $v_1(g_1) = v_2(g_1) = 1$ and 
$v_1(g_i) = v_2(g_i) = c$ for $i\in \{2, 3, \ldots, m-3\}$, as shown below, followed by the last 3 goods, $g_{m-2}, g_{m-1}, \allowbreak g_{m}$, such that $v_1(g_{m-2}) = v_1(g_{m}) = v_2(g_{m-1}) = v_2(g_{m}) = c^2$ and $v_2(g_{m-2}) = v_1(g_{m-1}) = c$.
The vertical line indicates the end of the initial view of algorithm $\mathcal{A}$ (initial item, $g_1$, and $m-4$ additional items). 
% \begin{table}[h!]
	\begin{center}  
	\begin{tabular}{lcccccc|ccc}
		& $g_1$ & $g_2$ & $g_3$   &\ldots & $g_{m-4}$ & $g_{m-3}$  & $g_{m-2}$ & $g_{m-1}$ & $g_{m}$ \\[5pt]
		agent 1: \hspace{1em} & $1$ & $c$ & $c$ &  \ldots & $c$ & $c$ & $c^2$ & $c$ & $c^2$ \\
		agent 2: \hspace{1em} & $1$ & $c$ & $c$ & \ldots & $c$ & $c$ &  $c$ & $c^2$ & $c^2$ 
	\end{tabular}
	\end{center}
% \end{table}
We are now going to argue similarly to the proof of Theorem \ref{thm:k-value_impossibility_lookahead}.

Due to symmetry, it is without loss of generality to assume that $g_1$ gets allocated to agent 1. We claim that once this happens, the algorithm must alternate between the two agents, giving all the odd-indexed goods up to (and including) $g_{m-1}$ to agent 1 and all the even-indexed ones to agent 2. 
To see this, suppose this is not the case. Then there are two consecutive goods given to the same agent; let $g_i, g_{i+1}$ be the first goods for which this happens. If $i$ is odd, then $g_i, g_{i+1}$ are both given to agent 1 and we have 
\[
v_2(A_1^{i+1}) = v_2(\{g_1, g_3, \ldots, g_i, g_{i+1}\})  
= v_2(\{g_2, g_4, \ldots, g_{i-1}\})  +c  +1 =  v_2(A_2^{i+1})  +c +1 \,.
\]
That is, the allocation $(A_1^{i+1},A_2^{i+1})$ is not \efo, contradicting the choice of $\mathcal{A}$. 
If $i$ is even, then $g_i, g_{i+1}$ are both given to agent 2. Note that this includes the extreme case where $i=m-2$; in this case the exponent $x$ below is equal to $2$, whereas in any other case it is equal to $1$: 
\begin{align*}
v_1(A_2^{i+1}) &= v_1(\{g_2, g_4, \ldots, g_i, g_{i+1}\}) = \frac{i\,c}{2}  +c^x  = \frac{i\,c}{(i-2)c+2}\cdot v_1(\{g_1, g_3, \ldots, g_{i-1}\})  +c^x\\ 
&   = \frac{i\,c}{(i-2)c+2}\cdot v_1(A_1^{i+1})  +c^x > v_1(A_1^{i+1})  +c^x \,
\end{align*}
where the third equality follows by explicitly calculating the value of $v_1(\{g_1, g_3, \ldots, g_{i-1}\}) = (i/2 -1) c +1$ and the last inequality follows from the fact that $i\,c>(i-2)c+2$, as $c>1$. Note that $c^x$ is the largest value agent~1 has for a good in $A_2^{i+1}$. 
As a result, the allocation $(A_1^{i+1},A_2^{i+1})$ is only $\frac{(i-2)c+2}{i\,c}$-\efo, again contradicting the choice of $\mathcal{A}$, since $\frac{(i-2)c+2}{i\,c}<1$ for $i \ge 2$. 
We conclude that algorithm $\mathcal{A}$ allocates the first $m-1$ goods so that $A_1^{m-1} = \{g_1, g_3, \ldots,g_{m-1}\}$ and $A_2^{m-1} = \{g_2, g_4, \ldots,g_{m-2}\}$.

Now, whoever gets $g_{m}$, the resulting allocation fails to be \efo. To see this, right before $g_{m}$ is allocated, we have 
\begin{align*}
v_1(A_1^{m-1}) &= \big(\frac{m}{2} - 1\big)\,c + 1 = \frac{(m-2)c + 2}{(m-4)c +2c^2} \Big[ \big(\frac{m}{2} - 2\big)\,c + c^2 \Big] \\ &= \frac{(m-2)c + 2}{(m-4)c +2c^2} \, v_1(A_2^{m-1}) < v_1(A_2^{m-1}) \,,   
\end{align*}
where the last inequality follows from the fact that $2c^2  > 2c + 2$ for $c > 2$, 
and
\begin{align*}
v_2(A_2^{m-1}) &= \big(\frac{m}{2} - 1\big)\,c = \frac{(m-2)c}{(m-4)c +2c^2+2} \Big[ \big(\frac{m}{2} - 2\big)\,c + c^2 + 1 \Big] \\ &= \frac{(m-2)c}{(m-4)c +2c^2 + 2} \, v_2(A_1^{m-1}) < v_2(A_1^{m-1}) \,,   
\end{align*}
where the last inequality follows from the fact that $2c^2  + 2> 2c$ for any $c$. In any case, by the time $g_{m}$ is given, the allocation fails to be \efo.

When $m$ is odd, the only difference is that the values of 
$g_{m-2}$ and $g_{m-1}$ are swapped, so that $v_1(g_{m-1}) =  v_2(g_{m-2}) =  c^2$ and $v_2(g_{m-1}) = v_1(g_{m-2}) = c$;
the analysis is essentially identical.

In the case of chores the construction is very similar. For instance, when $m$ is even, we have 
% \begin{table}[h!]
	\begin{center}  
	\begin{tabular}{lcccccc|ccc}
		& $g_1$ & $g_2$ & $g_3$   &\ldots & $g_{m-4}$ & $g_{m-3}$  & $g_{m-2}$ & $g_{m-1}$ & $g_{m}$ \\[5pt]
		agent 1: \hspace{1em} & $-1$ & $-c$ & $-c$ &  \ldots & $-c$ & $-c$ & $-c$ & $-c^2$ & $-c^2$ \\
		agent 2: \hspace{1em} & $-1$ & $-c$ & $-c$ & \ldots & $-c$ & $-c$ &  $-c^2$ & $-c$ & $-c^2$ 
	\end{tabular}
	\end{center}
% \end{table}
% \begin{table}[h!]
Notice how the values of $g_{m-2}$ and $g_{m-1}$ differ from their counterparts for goods; it is not just the sign, as their (absolute) values have been switched. The whole argument is the same as before, modulo the differences discussed in the corresponding part of the proof of Theorem \ref{thm:k-value_impossibility_lookahead}.
\end{proof}

%%%%%%%%%%%%%%%%%%%%%%%%%%%%%%%%%%%%%%
\section{A Matching-Inspired Framework}
\label{sec:main}
%%%%%%%%%%%%%%%%%%%%%%%%%%%%%%%%%%%%%%%

In light of the negative results of the previous section, it is natural to explore stronger online algorithm models that not only know part of the future but can manipulate it too. Augmenting our algorithms with buffers that can store and reorder some of the items is clearly an approach in this direction. Recall that this interpolates between online (buffer of size $0$) and offline (buffer of size $m$) fair division. We  focus on the problem of identifying a buffer size that is both reasonably small and allows us to always build temporal-\efo allocations. Importantly, \textit{can this be independent of $m$, the total number of items?} We resolve this question positively via our algorithm Store-and-Match (Algorithm \ref{alg:online-buffer}) for both modes of allocation we considered here, namely, the SAS model, where at most one item can be allocated per time step, and the DAB model, where any number of available items can be allocated at once.

\begin{theorem}
\label{thm:sas_positive}
    For any personalized $k$-value instance, Store-and-Match (Algorithm \ref{alg:online-buffer}) with a buffer of size $(n-1)k$ efficiently computes a temporal-\efo allocation in the SAS model. Moreover, this allocation is \ef every $n$ steps during the first $m-(n-1)k$ time steps in which an item is allocated.
\end{theorem}

Recall that a temporal-\ef/\efo allocation is \ef in every time step but the last one (where it is \efo).

\begin{theorem}
\label{thm:dab_positive}
    For any personalized $k$-value instance, Store-and-Match (Algorithm \ref{alg:online-buffer}) with a buffer of size $(n-1)k$ efficiently computes a temporal-\ef/\efo allocation in the DAB model.
\end{theorem}

% % % % % % % % % % % % % % % % % %
\subsection{The Illustrative Case of Two Agents} \label{sub:positive_2_agents}
% % % % % % % % % % % % % % % % % %

To illustrate the high-level idea behind our main technical result, which leads to Theorems \ref{thm:sas_positive} and \ref{thm:dab_positive}, we first present it for the easier case of two agents. 
Intuitively,  Algorithm \ref{alg:online-buffer} works as follows: Suppose we have already constructed a partial EF allocation using a subset of $M$. Once the buffer fills up,\footnote{This requirement is for presentation purposes and is not crucial. The algorithm could have been designed to allocate before the buffer is full if suitable pairs are found, without violating temporal-EF1 or EF every two steps. Similarly for $n\ge3$.} we construct a bipartite graph between the agents and all currently available items (i.e., the items in the buffer $B$ and the incoming item). We refer to these items as \emph{live} items, and there is an edge in the graph connecting each agent \textit{to her most valuable items}. Then, if there is a perfect matching between the two agents and two of the live items, these items can be allocated to the agents, and the allocation remains \ef. Otherwise, there is a contested item, which we temporarily hide and repeat the same process. We eventually show that a buffer of size $k$ suffices to always have a compatible pair of items to allocate in an envy-free manner. Hence, by starting with an EF partial allocation, we can maintain EF every two time steps, while in the intermediate step the allocation is EF1.

To formally analyze the algorithm we will introduce some tools. Up to Lemma \ref{lem:perfect_matching_exists}, we state and prove everything for the general case, as there is no particular difference between $n=2$ and $n\ge 3$. 

\begin{definition}[\perfect]
Let $G=(L,R,E)$ be a bipartite graph and let $\mu \subseteq E$ be a matching in $G$. 
We say that $\mu$ is \emph{left-perfect} if every vertex $v \in L$ is incident to exactly one edge in $\mu$.
\end{definition}

We will use the graph-theoretic version of Hall's Theorem~\citep{Hall35}. 

\begin{theorem}[\cite{Hall35}]\label{thm:Hall}
Let $G = (L, R, E)$ be a bipartite graph. For any subset $S \subseteq L$, define the \emph{neighborhood} of $S$, 
% \[
$\Gamma(S) = \{\, r \in R \mid \exists\, \ell \in S \text{ such that } (\ell,r) \in E \,\}$.
% \]
Then there exists a \perfect  if and only if
%\[
$|\Gamma(S)| \ge |S| ~~ \text{for every } S \subseteq L.$
%\]
\end{theorem}

Although in our model we have defined cardinal preferences for the agents, it suffices to use only the weaker form of ordinal preferences induced by the cardinal form. Let $\succeq_i$  be the weak ordering over the items in $M$ induced by $v_i$, where for all $g, g' \in M$, $g \succeq_i g'$ if and only if $v_i(g) \ge v_i(g')$. 

For any $S \subseteq M$ and any $i \in N$, we define the \textit{top set}, $T_i(S)$, as the set of the most preferred items in $S$ according to $\succeq_i$: 
% \[
$T_i(S) = \{ g \in S : g \succeq_i g' \text{ for all } g' \in S \}$.
% \]
Note that there can be multiple such items, and also that $T_i(S)\neq \emptyset$ for any $S\neq \emptyset$.

The following definition provides a very useful graph structure we will use throughout our proof.
\begin{definition}
Given the set of agents $N$ and a set of items $S$, we denote by $G(S)$ the bipartite graph between $N$ and $S$ where an edge $(i,g)$ exists if and only if $i \in N$, $g \in S$ and $g \in T_i(S)$. We refer to $G(S)$ as the \emph{top choice graph} with respect to $S$. 
\end{definition}
Note that in $G(S)$ every $i\in N$ has degree at least one, since $T_i(S)\neq \emptyset$. As a first step, the following simple lemma says that a \perfect in the above graph is enough to expand an EF partial allocation, maintaining envy-freeness. 

\begin{lemma}\label{lem:A_perfect_ matching_is_enough_for_EF}
    Let $\mathcal{A}=(A_1,\dots,A_n)$ be an EF partial allocation over a set $N$ of agents, $S$ be a set of items, and let $G(S)$ be the induced top choice graph.  If there exists a \perfect $\mu:N\rightarrow S$ in $G(S)$, 
    then the (partial) allocation $\mathcal{A}'=(A'_1,\dots,A'_n)$ such that $A_i' = A_i \cup \{\mu(i)\}$ is EF.
\end{lemma}

\begin{proof}
    By assumption, $v_i(A_i)\geq v_i(A_j)$ for all $i,j \in N$. Due to the \perfect, $v_i(\mu(i))\geq \allowbreak v_i(g)$ for all $g \in S$. Hence, for any $i,j \in N$:
    % \begin{align*}
        $v_i(A_i \cup \{\mu(i)\}) = v_i(A_i)  + v_i(\mu(i)) \geq v_i(A_j) + v_i(\mu(j)) = v_i(A_j \cup \{\mu(j)\})$
    % \end{align*}
    and the lemma follows.
\end{proof}

\begin{algorithm}[t]
	\setstretch{1.1}
	\caption{Store-and-Match$(v_1, \ldots, v_n; M; b)$ \\{\small {($M$ and $v_i$, $i\in [n]$, are revealed in an online fashion, one item at a time; $b$ is the capacity of the buffer)}}}
	\label{alg:online-buffer}
	\begin{algorithmic}[1]
		\State $B \gets \emptyset$ \Comment{initialization of the buffer}
		\State $Q \gets \emptyset$ \Comment{set of scheduled allocations (for the SAS model)}
		
		\For{each arriving item $g$}
		
		\If{$Q \neq \emptyset$}
		\State Allocate one item from $Q$ to its agent and delete it from $B$ \label{line:Q_non-empty} \Comment{choose lexicographically}
		\EndIf    \vspace{2pt}
		\If{$|B|< b$} \Comment{if the buffer is not full, we store the item}
		\State $B \gets B \cup \{g\}$
		\Else  \Comment{the buffer is full}
		\State Let $R \gets B \cup \{g\}$ \label{line:k+1}
		\State Construct the top choice graph $G(R)$, where $(i,h) \in E$ if and only if $h \in T_i(R)$
		\While{$G(R)$ does not admit a \perfect  $\mu$} \label{line:while}
		\State Let $H\subseteq R$ be the items of $\Gamma(X)$ adjacent to at least two agents of $X$, for some $X\subseteq N$ with $|\Gamma(X)|<|X|$ in $G(R)$ \label{line:Hall}
		\State $R \gets R\setminus H$ 
		\EndWhile
		\State Add the agent-item pairs matched by $\mu$ to $Q$ \Comment{$\mu$ is guaranteed to exist at this point}
		\State For the SAS model: Allocate one item from $Q$ to its agent and delete it from $B$ \label{line:allocate_matching} 
		\State For the DAB model: Allocate all items from $Q$ and delete them from $B$ \label{line:allocate_matching_DAB} 
		\If{$g$ has not been allocated} 
		\State $B \gets B \cup \{g\}$ \Comment{store $g$, if needed}
		\EndIf
		\EndIf
		\EndFor\vspace{2pt}
		
		\State \textbf{Finalization:} \Comment{the stream ends but $
			|B|>0$}
		\While{$Q \neq \emptyset$}
		\State Allocate one item from $Q$ to its agent and delete it from $B$ \label{line:last_matching}
		\EndWhile
		\State Allocate the remaining items in $B$ via the Double Round-Robin algorithm of~\citet{aziz2022fair} \label{line:DRR}
	\end{algorithmic}
\end{algorithm}

The following is the key technical lemma behind our positive results, and also forms the main difference in the analyses of the case of $2$ agents and the general case of $n\ge 3$. It shows that, assuming a personalized $k$-value instance,  among $k+1$ available items  we can always find a pair to expand the allocation and keep it EF. Note that this is not the simplest way to show Lemma \ref{lem:perfect_matching_exists} for $n=2$ but has the advantage of providing a clean picture of the main idea that is used for general $n$.

\begin{lemma}\label{lem:perfect_matching_exists}
Assume we have a personalized $k$-value instance with two agents. Let $(A_1,A_2)$ be an EF-partial allocation and $S$ be any set of $k+1$ unallocated items. Then there exists at least one pair of items  $\{a,b\} \subseteq S$ such that $(A_1\cup \{a\}, A_2 \cup \{b\})$ is an EF-partial allocation.
\end{lemma}

\begin{proof}
    We prove the lemma via a bounded progress argument. We construct a sequence of top choice graphs 
    $G(R_\ell)$, where
    the sets $R_\ell$ form a decreasing sequence, that is, $R_{\ell+1} \subsetneq R_\ell$, and $R_1 = S$. At each step, we examine the graph $G( R_\ell)$. If it admits a \perfect, we have found an EF partial allocation due to Lemma~\ref{lem:A_perfect_ matching_is_enough_for_EF} and we may stop. Otherwise, we remove from $R_{\ell}$ the contested vertices (items)  
    and proceed to the next graph in the sequence. We will show that this process can continue for up to $k$ steps and that for some $j \leq k$, the graph $G(R_j)$ admits a \perfect.

    Let us start with the first iteration of this process and the top choice graph $G(R_1)$, where $R_1=S$. If $G(R_1)$ has a left-perfect matching, we are done by Lemma~\ref{lem:A_perfect_ matching_is_enough_for_EF} and terminate. Suppose now that $G(R_1)$ does not have a left-perfect matching. Then, due to the fact that $|N|=2$ (and that all agents in $N$ have at least one neighbour), any set violating the condition of Hall's Theorem (Theorem \ref{thm:Hall}) has a simple form: both agents in $N$ are connected to the same item, which is the most preferred for both agents in the set $R_1$ and they are not connected to any other items (otherwise a perfect matching would exist). Let $h$ be this contested item. Define $R_2 = R_1\setminus \{h\}$ and proceed to the next iteration, where the top choice graph is $G(R_2)$.
       
    By repeating the above reasoning, whenever we reach iteration $\ell>1$, either a left-perfect matching exists in $G(R_\ell)$, in which case we are done, or both agents are connected only to one item, which is their most preferred in the set $R_\ell$.  
Also, in order to reach iteration $\ell$, it means that for all $\lambda<\ell$ the graph $G(R_\lambda)$ does not admit a left-perfect matching and we have deleted the $\ell-1$ most preferred items of both agents from the initial set $S$.
In particular, for any $\ell_1 < \ell_2 < \ell$, the $\ell_1$-th item we deleted was strictly preferred over the $\ell_2$-th one for both agents.

Since $S$ has exactly $k+1$ items and the agents have $k$-value valuations, there are at most $k$ possible different ranking positions for the items. Thus, if the process reaches $G(R_k)$ (having deleted $k-1$ items), then the two remaining items have the $k$-th highest value for both agents. But this means that a \perfect exists in $G(R_k)$. Therefore, when $|S|=k+1$, we are always able to find a left-perfect matching after at most $k$ iterations,  and this concludes the proof. 
\end{proof}

We are now ready to prove Theorems \ref{thm:sas_positive} and \ref{thm:dab_positive} for the case of two agents. 
Nevertheless, for brevity, we present only the proof of Theorem~\ref{thm:dab_positive} here. Both theorems are proved for general $n$ in the next section.

\begin{proof}[Proof of Theorem \ref{thm:dab_positive} for $n=2$] 
Consider the times during which the algorithm allocated some item(s) to some agent(s). This can occur in  lines \ref{line:allocate_matching_DAB} and \ref{line:DRR} of Algorithm~\ref{alg:online-buffer}. In particular, for as long as new items keep arriving, Algorithm~\ref{alg:online-buffer} allocates items in line \ref{line:allocate_matching_DAB} based on the left-perfect matchings it computes. As each of these matchings is allocated at once, the allocation remains \ef in every step until items stop arriving. Once this happens, an \efo partial allocation is computed offline via Double Round-Robin (by Theorem \ref{thm:aziz_DDR_is_EF1}) and is added in a single time step to the existing \ef partial allocation, resulting in an \efo final allocation~(for any two agents, the item whose removal eliminates the envy in the \efo part can also be removed from the final bundles, since the existing part is \ef). Overall, the allocation computed by Algorithm~\ref{alg:online-buffer} is temporal-\ef/\efo. Moreover, every step of the computation  is  clearly done in polynomial time.
\end{proof}

% % % % % % % % % % % % % % % % % %
\subsection{Any Number of Agents} \label{sub:positive_general_n}
% % % % % % % % % % % % % % % % % %

In this section we analyze Algorithm~\ref{alg:online-buffer} for an arbitrary number of agents. The analysis follows Section \ref{sub:positive_2_agents} at a high level; however, we need the generalized analog of Lemma~\ref{lem:perfect_matching_exists}.

\begin{lemma}\label{lem:perfect_matching_exists_for_n_agents}
    Let $S$ be a set of $(n-1)k+1$ items. For a set $N$ of $n$ agents with personalized $k$-value valuations with an EF-partial allocation $(A_1,\dots,A_n)$, there exists a matching $\mu:N\rightarrow S$ such that the allocation $\mathcal{A}'$ with $A'_i=A_i\cup \{\mu(i)\}$, for all $i \in N$, is an EF-partial allocation.
\end{lemma}

\begin{proof}%[\textbf{Proof of Lemma \ref{lem:perfect_matching_exists_for_n_agents}}]
    We will prove the lemma with a process similar to Lemma~\ref{lem:perfect_matching_exists}. The arguments follow the same structure but now we need to work  more carefully with the sets that violate the condition of Hall's Theorem (i.e., subsets $X \subseteq N$ for which $|\Gamma(X)| < |X|$).

    Again, we construct a sequence of top choice graphs $G(R_\ell)$ where the sets $R_\ell$ form a decreasing sequence, that is, $R_{\ell+1} \subsetneq R_\ell$, and $R_1 = S$. At each step, we examine the graph $G(R_\ell)$. If it admits a \perfect we have found an EF partial allocation due to Lemma~\ref{lem:A_perfect_ matching_is_enough_for_EF} and we terminate. Otherwise, we remove from $R_{\ell}$ some highly contested vertices / items (see below)  and proceed to the next graph in the sequence. We will show that this process can continue for up to $(n-1)k$ steps and that for some $\ell \leq (n-1)k$ the graph $G(R_\ell)$ admits a \perfect.

    It remains to show that such a \perfect must eventually be found. Suppose that, for some $\ell$, the graph $G(R_{j})$ does not admit a \perfect, for all $j\in [\ell]$. Specifically for $G(R_{\ell})$, by Hall's Theorem (Theorem~\ref{thm:Hall}), there exists a set $X\subseteq N$ such that
$|\Gamma_\ell(X)|<|X|$, 
where $\Gamma_\ell(X)$ is the set of neighbors of $X$ in $G(R_\ell)$, as defined in the statement of Theorem~\ref{thm:Hall}. That is, $\Gamma_\ell(X) \subseteq R_\ell$. 
Note that such a set $X$ can be found in polynomial time, as it reduces to a reachability problem \citep{AMNS17}. Specifically, we compute a maximum matching $\mu$ in $G(R_\ell)$ and let $X$ be the set of agents that are reachable from the agents left unmatched by $\mu$ via $\mu$-alternating paths (i.e., paths that start with an edge not in $\mu$ and alternate between edges not in $\mu$ and edges in $\mu$). Every item of $\Gamma_\ell(X)$ is matched by $\mu$, as otherwise there would be an augmenting path, and its partner is an agent of $X$; conversely, every matched agent of $X$ is matched to an item of $\Gamma_\ell(X)$. Thus, $|\Gamma_\ell(X)|$ is equal to $|X|$ minus the number of unmatched agents, which is positive since $G(R_\ell)$ admits no \perfect, i.e., $|\Gamma_\ell(X)|<|X|$. We note that the argument below works for any $X\subseteq N$ with $|\Gamma_\ell(X)|<|X|$. So, assume that we have such an $X$.

    At this point, we partition the vertices in $\Gamma_\ell(X)$ into two types. An item is \textit{uniquely demanded} if it is adjacent to exactly one agent of $X$, and \textit{$X\!$-contested} if it is adjacent to at least two agents of $X$. 
    
    Let $H_\ell$ be the set of $X\!$-contested items in $\Gamma_\ell(X)$. Observe that the set $H_\ell$ is non-empty. Indeed, every agent in $X$ is adjacent to at
least one item in $\Gamma_\ell(X)$. If every item in $\Gamma_\ell(X)$
were uniquely demanded, then the uniquely demanded items would be enough to give a distinct item
to every agent in $X$, contradicting the fact that $|\Gamma_\ell(X)|<|X|$. Hence, there exists at least one $X\!$-contested item.
    
    We hide from $R_\ell$ the set of $X\!$-contested items, 
    thus creating the set $R_{\ell+1} = R_\ell\setminus   H_\ell$. Note that, in every iteration that fails to find a \perfect the set of remaining items strictly shrinks.

    We now identify the agents whose current set of most preferred items (recall that we refer to these as the agent's \emph{top set}) disappears after hiding
$H_\ell$. Let
\[
    L_\ell=\{i\in X:T_i(R_\ell)\cap R_{\ell+1}=\emptyset\}\,.
\]
These are exactly the agents in $X$ for whom all items in their current top
set are hidden in this step. Hence, after passing from $R_\ell$ to
$R_{\ell+1}$, every agent $i$ in $L_\ell$ has a strictly lower current best
remaining value than what they had in $R_{\ell}$. 
We claim that
 \[ |H_\ell|<|L_\ell|\,.\]  
Indeed, consider an agent $i\in X\setminus L_\ell$ whose top set survives. By the definition of
$L_\ell$, there exists an item
\begin{align*}
    g_i\in T_i(R_\ell)\setminus H_\ell \,.
\end{align*}
Since $g_i\in \Gamma_\ell(X)\setminus H_\ell$, it is uniquely demanded
with respect to $X$. Therefore no two different agents in $X\setminus L_\ell$
can choose the same such item. Hence
\begin{align*}
|\Gamma_\ell(X)\setminus H_\ell|
    \geq |X\setminus L_\ell|.
\end{align*}
Intuitively, the set of items shrinks faster than the set of agents.
Then we can write, $ \Gamma_\ell(X)=H_\ell \cup (\Gamma_\ell(X)\setminus H_\ell)$. Combining Hall's violation, i.e., that $|\Gamma_\ell(X)|<|X|$, with the above inequality, we get
\begin{align*}
    |H_\ell|
    =
    |\Gamma_\ell(X)|-|\Gamma_\ell(X)\setminus H_\ell|
    <
    |X|-|X\setminus L_\ell|
    =
    |L_\ell|.
\end{align*}
Thus, in every iteration that fails to produce a \perfect the number of hidden items is strictly smaller
than the number of agents whose current top set is completely removed.

We next show that $|R_\ell|\ge n$ for every $\ell$ such that $R_\ell\neq\emptyset$. Fix such an $\ell$ and, for each agent $i\in N$, let $\lambda_i$ be the number of indices $j<\ell$ such that $i\in L_j$. Since the sets $R_1\supseteq R_2\supseteq \cdots\supseteq R_\ell$ are non-empty, the largest value agent $i$ has for an item of the current set never increases, and it strictly decreases from $R_j$ to $R_{j+1}$ whenever $i\in L_j$. As agent $i$ has at most $k$ distinct values, we have $\lambda_i\le k-1$; clearly, we also have $\lambda_i\le \ell-1$. Therefore,
\begin{align*}
|R_\ell| &= |S| - \sum_{j=1}^{\ell-1}|H_j| \ge (n-1)k+1 - \sum_{j=1}^{\ell-1}(|L_j|-1) = (n-1)k+\ell - \sum_{i\in N}\lambda_i \\ &\ge (n-1)k+\ell - (n-1)(k-1) - (\ell-1) = n\,,
\end{align*}
where the first inequality follows from $|H_j|<|L_j|$, the second equality holds because $\sum_{j=1}^{\ell-1}|L_j| = \sum_{i\in N}\lambda_i$, and the last inequality follows by using the bound $\lambda_i\le k-1$ for $n-1$ of the agents and the bound $\lambda_i\le \ell-1$ for the remaining one.

As a result, the process never runs out of items: if $R_\ell\neq\emptyset$ and $G(R_\ell)$ admits no \perfect, then $|H_\ell|\le|\Gamma_\ell(X)|<|X|\le n\le |R_\ell|$, and thus $R_{\ell+1}=R_\ell\setminus H_\ell\neq\emptyset$. Since the sets $R_\ell$ are strictly decreasing and always contain at least $n$ items, after at most $|S|-n=(n-1)(k-1)$ iterations we reach a graph $G(R_\ell)$ that admits a \perfect. We conclude that with $(n-1)k+1$ items, we can guarantee an EF allocation 
for $n$ agents in a personalized $k$-value instance.
\end{proof}

Here we shall briefly discuss the differences between the proof of Lemma \ref{lem:perfect_matching_exists_for_n_agents} 
and that of Lemma \ref{lem:perfect_matching_exists}. A first issue is which live items to hide when $G(R)$ admits no \perfect. Given a set $X\subseteq N$ violating the condition of Theorem~\ref{thm:Hall}, i.e., $|\Gamma(X)|<|X|$, we hide only the items of $\Gamma(X)$ that are adjacent to at least two agents of $X$ (line~\ref{line:Hall}); this guarantees that in each such step fewer items are hidden than there are agents whose entire top set is hidden. A second matter is how to efficiently compute such a set $X$. Thankfully, it turns out that this can be reduced to an easy reachability problem: if $\mu$ is a maximum matching in $G(R)$, the set $X$ of agents reachable from the agents left unmatched by $\mu$ via $\mu$-alternating paths satisfies $|\Gamma(X)|<|X|$ \citep{AMNS17}, as described in the proof above. A last subtle point is that the top choice graph does not directly imply an upper bound on the values of all live items. For $n=2$, we know that after $\ell$ failures to find a matching we have deleted the $\ell$ most preferred items of both agents. For $n\ge3$, however, each failure to find a matching might affect only a subset of agents, yet we may delete multiple most preferred items for each one of them. To deal with this challenge we resort to careful counting arguments that generalize the simple idea of hiding a single contested item at a time.

Once we do have Lemma \ref{lem:perfect_matching_exists_for_n_agents}, the proofs of Theorems \ref{thm:sas_positive} and \ref{thm:dab_positive} are not particularly hard, especially the latter.

\begin{proof}[\textbf{Proof of Theorem~\ref{thm:sas_positive}}]
We consider the times during the execution of Algorithm~\ref{alg:online-buffer} where an item is allocated to some agent. This can occur in lines \ref{line:Q_non-empty}, \ref{line:allocate_matching}, \ref{line:last_matching}, and \ref{line:DRR}.

Let us examine first the assignments that take place throughout the execution of the for loop of Algorithm~\ref{alg:online-buffer}, in lines \ref{line:Q_non-empty} and \ref{line:allocate_matching}, that is, as long as new items keep arriving. 
During these iterations, Algorithm~\ref{alg:online-buffer} allocates items based on the left-perfect matchings it produces. Suppose that throughout the for loop, it gets to produce the matchings $\mu_1, \mu_2,...,\mu_\ell$ (in that order). 
Since we can allocate at most one item per step, the $n$ items of a \perfect cannot be allocated simultaneously, and thus they are inserted into $Q$ and allocated one by one in lexicographic order, before the next \perfect is found. This means that the sequence of allocation decisions is exactly 
\[
    \mu_1(1),...,\mu_1(n),\mu_2(1),...,\mu_2(n),...,\mu_\ell(1),...,\mu_\ell(n) \,.
\]
Note that  we start with the empty allocation, which is EF. Since in line \ref{line:k+1}, we examine $(n-1)k+1$ items (namely the items in the buffer plus the newly arrived item), Lemma  \ref{lem:perfect_matching_exists_for_n_agents} ensures that a \perfect will be constructed during the while loop that starts at line \ref{line:while}.
Furthermore, every \perfect is constructed on the graph $G(R)$, and thus Lemma~\ref{lem:A_perfect_ matching_is_enough_for_EF} implies that after all $n$ items of the matching are allocated, the allocation remains EF. This means that during the execution of the for loop of Algorithm~\ref{alg:online-buffer}, the partial allocations we  construct are EF every $n$ time steps. 
For the intermediate steps, 
where we allocate the first $n-1$ items of each matching, the allocation is trivially EF1, since it was EF right before and each agent gets exactly one new item; 
indeed, consider any two agents $i, j$ and let $\mu$ be the current matching. If both $\mu(i)$ and $\mu(j)$ have been allocated, then $i$ does not envy $j$ by (the proof of) Lemma~\ref{lem:A_perfect_ matching_is_enough_for_EF}. If exactly one of $\mu(i)$ and $\mu(j)$ has been allocated, then removing it re-establishes the envy-freeness of the allocation right before $\mu$ started being allocated.
Hence, these allocations are always EF1.

Consider now the assignments that take place in line \ref{line:last_matching}, after the for loop is over. These follow the same reasoning as before, since they concern the allocation of the items that belong to the last produced matching, $\mu_{\ell}$. 
Therefore, once the second while loop terminates, right before line \ref{line:DRR}, the current allocation is EF.
Thus, \textit{so far}, the allocation is temporal-\efo and \ef every $n$ steps for the first $m-(n-1)k$ time steps that an item is assigned.

Finally, the remaining (at most $(n-1)k$) items that sit in the buffer, are allocated using the Double Round-Robin algorithm of~\citet{aziz2022fair}. 
About these last $|B|$ items, first note that they may not admit a matching at all. Therefore, the periodic guarantee of EF is not relevant beyond this point. Second, by Lemma~\ref{lem:DRR_is_TEF1} in Appendix~\ref{sec:DRRisTEF1}, Double Round-Robin allocates the items sequentially, maintaining an \efo allocation throughout its execution. Since the allocation before that was EF and we essentially run Double Round-Robin offline, it follows that the allocation remains temporal-EF1 throughout this phase as well: for any two agents, the item whose removal eliminates the envy in the part allocated by Double Round-Robin also eliminates the envy in the whole allocation, because the part allocated before is EF. 

It is not hard to see that all the steps involved---allocating items from $Q$, constructing $G(R)$, finding at most $(n-1)k$ maximum cardinality matchings in $G(R)$, finding $H$, updating $R$ and $Q$, running Double Round-Robin and allocating the last items according to its output---run in polynomial time. Hence, Algorithm~\ref{alg:online-buffer} runs in polynomial time overall.
\end{proof}

% \smallskip 

\begin{proof}[\textbf{Proof of Theorem~\ref{thm:dab_positive} for general $\bm{n}$}]
This is a simpler version of the proof of Theorem \ref{thm:sas_positive} and is nearly identical to the proof we gave for the case of two agents. First, notice that now an item may be allocated to an agent only in lines \ref{line:allocate_matching_DAB} and \ref{line:DRR} of Algorithm~\ref{alg:online-buffer}. In particular, for as long as new items keep arriving, Algorithm~\ref{alg:online-buffer} allocates items in line \ref{line:allocate_matching_DAB} based on the left-perfect matchings it computes and which are guaranteed to exist as long as there are enough items by Lemma \ref{lem:perfect_matching_exists_for_n_agents}. Each of these matchings is allocated at once, thus the allocation remains \ef in every step until items stop arriving, say at time $\tau$. When this happens, we run Double Round-Robin offline on the items left in $B$. This results in a potential (i.e., still unallocated) \efo partial allocation $\mathcal{B} = (B_1, \ldots, B_n)$ of all the items in $B$ to the agents.
Adding $\mathcal{B}$ in a single time step to the existing \ef partial allocation, results in the \efo final allocation $(A^{\tau}_1\cup B_1, \ldots, A^{\tau}_n\cup B_n)$, by the same argument as in the proof of Theorem~\ref{thm:sas_positive}. Overall, the allocation computed by Algorithm~\ref{alg:online-buffer} is temporal-\ef/\efo. Moreover, every step of the computation  is  done in polynomial time as we argued in the proof of Theorem \ref{thm:sas_positive}.
\end{proof}

% % % % % % % % % % % % % % % % % %
\subsection{Limitations on the Power of Buffers} \label{sub:impossibility_buffer}
% % % % % % % % % % % % % % % % % %
A natural question at this point is whether one could do equally well using smaller buffers. In the case of the DAB model and of temporal-\ef/\efo the answer is no and this is already hinted at in the proof of Lemma \ref{lem:perfect_matching_exists_for_n_agents}. The next theorem makes this explicit, also showing that our Theorem \ref{thm:dab_positive} is tight.

\begin{theorem}\label{thm:buffer-impossibility-batches}
Let $k\ge 2$ be an integer and $\mathcal{A}$ be a deterministic online algorithm in the DAB model with a buffer of size $(n-1)k - 1$. Then $\mathcal{A}$ may fail to produce an \ef allocation in at least half of the time steps it updates the allocation for $k$-value instances, even if all items are goods or chores. 
\end{theorem}

\begin{proof}%[\textbf{Proof of Theorem \ref{thm:buffer-impossibility-batches}}]
Suppose we have a deterministic allocation algorithm $\mathcal{A}$ with a buffer of size $b = (n-1)k - 1$.
We are going to construct a sequence of items that forces $\mathcal{A}$ to fail to produce an \ef allocation in at least half of the time steps where the allocation is updated.
This sequence will use $k$ distinct values, thus implying the bound of the statement. We are going to show the result for goods; the proof for chores is essentially identical, the only difference being the signs of all the values involved.

First, consider the goods $g_1, g_2, \ldots, g_{b+1}$, in this order, such that $v_i(g_j) = n^{\ell - 1}$, where $\ell = \lceil j / (n-1) \rceil$, for $i\in N, j\in [b+1]$, as shown below. For the sake of presentation, the goods are grouped in groups of size $n-1$ as follows: $G_{\ell}=\{g_{(\ell-1)(n-1)+1},\ldots,g_{\ell(n-1)}\}$, for $\ell \in [k],$ i.e., $G_1$ contains the first $n-1$ goods of value $1$ each, $G_2$ contains the next $n-1$ goods of value $n$ each, and so on. 
% \begin{table}[h!]
	\begin{center}  
	\begin{tabular}{lcccccc}
		& $g\in G_1$ & $g\in G_2$ & $g\in G_3$  &\ldots &  $g\in G_{k-1}$ & $g\in G_k$  \\[5pt]
		agent 1: \hspace{0.3em} & $1$ & $n$ & $n^2$ & \ldots &   $n^{k-2}$ & $n^{k-1}$  \\
		agent 2: \hspace{0.3em} & $1$ & $n$ & $n^2$ & \ldots &   $n^{k-2}$ & $n^{k-1}$  \\
        \quad $\vdots$ & $\vdots$ & $\vdots$ & $\vdots$ & $\ddots$ & $\vdots$ & $\vdots$  \\[4pt]
        agent $n$:\hspace{0.3em}& $1$ & $n$ & $n^2$ & \ldots &    $n^{k-2}$ & $n^{k-1}$ 
	\end{tabular}
	\end{center}
% \end{table}

We begin with two simple observations; first, there are not enough copies of each value for everyone and,  second, even all goods of value less than $n^{\ell - 1}$ do not sum up to this value for any $\ell \in [k]$. To see the latter:
\[(n-1) \sum_{i=0}^{\ell - 2} n^i = (n-1) \frac{n^{\ell - 1} - 1}{n-1} = n^{\ell - 1} - 1< n^{\ell - 1} \,.\]
As a result of these simple facts, it is impossible to construct an \ef allocation by using \textit{any} non-empty subset of $\{g_1, g_2, \ldots, g_{b+1}\}$: some agent receives no good of the largest value present, so her bundle is worth less than that value, whereas the bundle of some other agent is worth at least that much. 

Now, let $\tau_1, \tau_2, \ldots$ be the time steps during which $\mathcal{A}$ allocates at least one good. We have that $\tau_1 \le b+1$, since by the end of time step $b+1$ the algorithm must have allocated at least one good. No matter what the allocation $\mathcal{A}^{\tau_1}$ is, it fails to be \ef. 

From this point on, whenever algorithm $\mathcal{A}$ allocates one or more goods, the stream of goods replaces them (in arbitrary order) with identical copies, so that the set of items available (i.e., the goods in the buffer plus the next good on the sequence) is always equivalent to a subset of $\{g_1, g_2, \ldots, g_{b+1}\}$. 

Suppose that $\mathcal{A}^{\tau_s}$ is \ef for some $s\ge 2$. Since the valuation functions of all agents are identical, this means that the allocation $\mathcal{A}^{\tau_s}$ is equitable, i.e., all agents have the exact same value for all allocated bundles. Then, however, $\mathcal{A}^{\tau_{s+1}}$ cannot be \ef as it is augmented by (essentially) a non-empty subset of $\{g_1, g_2, \ldots, g_{b+1}\}$. Hence, $\mathcal{A}^{\tau_1}$ is not \ef and no two consecutive updates both result in \ef allocations, so at least half of the updates fail to produce an \ef allocation.
\end{proof}

When one turns to the SAS model, Theorem \ref{thm:buffer-impossibility-batches} does not have any nontrivial implications. On the one hand, the SAS model is much weaker in how it allocates the items \textit{but so is the benchmark} we need to compare against, namely, being \efo rather than \ef in every step. Nevertheless, we show that, even in the case of two agents, a buffer of size roughly $\sqrt{k/6}$, where $k$ is the number of distinct values, is not sufficient for computing  $\beta$-temporal-\efo allocations, for any $\beta \in (0,1]$. To achieve that, we need a novel recursive construction of a highly nontrivial adversary and an equally delicate analysis. We view this impossibility result as one of the technical highlights of this work.

    \begin{theorem}\label{thm:buffer-impossibility-single_item}
		Let $k\ge 3$ be an integer, $\beta \in (0,1]$, and $\mathcal{A}$ be a deterministic online algorithm in the SAS model  that uses a buffer of size $\lfloor\sqrt{(k - 3)/6}\rfloor$. 
        Then $\mathcal{A}$ cannot always maintain a $\beta$-temporal-\efo allocation for more than $5k/6$ time steps for $k$-value instances, even if there are only two agents and all items are goods or chores.  
	\end{theorem}

\begin{proof}%[\textbf{Proof of Theorem \ref{thm:buffer-impossibility-single_item}}]
		Suppose we have a deterministic allocation algorithm $\mathcal{A}$ with a buffer of size $b$.
		We are going to construct an adversary $\texttt{Adv}_b$ who generates a sequence of items that forces $\mathcal{A}$ to fail to produce a $\beta$-\efo allocation within at most $5b^2 + 3$ time steps. Moreover, this sequence will use at most $6b^2 + 3$ distinct values, thus implying the bound of the statement. We are going to show the result for goods; the proof for chores is completely analogous.
		
        Let $c\ge \max\{3, \lceil \beta^{-1} \rceil +1\}$ and notice that $c$ is such that $(c-1)^{-1} \le \beta$, so it suffices to show that $\mathcal{A}$ fails to produce a $1/(c-1)$-\efo.
        The proof is by induction on the buffer size $b$; hence, our construction is going to be recursive. For $b = 0$, consider the adversary $\texttt{Adv}_0(c)$, parameterized by $c$, who generates a stream of goods that begins with $g_1$, such that $v_1(g_1) = v_2(g_1) = 1$. It is without loss of generality to assume that $\mathcal{A}$ assigns $g_1$ to agent 1; if not, our adversary may switch the values the items have for agents 1 and 2. Then, the following items are $g_2, g_3$, such that $v_1(g_2) = v_1(g_3) = c$, $v_2(g_2) = 1/c$ and $v_2(g_3) = 1$, as shown below.
		% \begin{table}[h!]
			\begin{center}  
				\begin{tabular}{lcccc}
					& $g_1$ & $g_2$ & $g_3$ &\ldots \\
					agent 1: \hspace{1em} & \scircled{1} & c &c &  \ldots  \\
					agent 2: \hspace{1em} & 1 & \circled{1/c} & 1 &  \ldots 
				\end{tabular}
			\end{center}
			Given that $g_1$ is added to $A_1$, either $g_2$ is also added to $A_1$ and the resulting allocation $(\{g_1, g_2\}, \emptyset)$ is not $1/(c-1)$-\efo from the point of view of agent 2, or $g_2$ is added to $A_2$; we assume the latter.
			Now, however, no matter who receives $g_3$, the resulting allocation is only $1/c$-\efo and algorithm $\mathcal{A}$ fails to maintain a $1/(c-1)$-\efo allocation within the first 3 time steps, using 3 distinct values, for any $c\ge 3$. Further, all values used are at most $c^{0+1}$.
			
			So, for any algorithm $\mathcal{A}$ with a buffer of size $b-1\ge 0$ and for any $c\ge 3$, assume that there is an adversary $\texttt{Adv}_{b-1}(c)$ that forces $\mathcal{A}$ to produce an  allocation which is not $1/(c-1)$-\efo within at most $5(b-1)^2 + 3$ time steps by using at most $6(b-1)^2 + 3$ distinct values and these values are upper bounded by $c^{b}$.
			
			For our inductive step, assume that algorithm $\mathcal{A}$ has a buffer of size $b\ge 1$ and fix any $c\ge 3$. 
			Consider the adversary $\texttt{Adv}_b(c)$ who generates an initial stream of goods $g_0, g_1, \ldots, g_{b}$, such that 
			$v_1(g_i) = v_2(g_i) = 1/c^{4 (b+2) i} := \varepsilon^i$ for $i\in \{0, 1, \ldots, b\}$, as shown below, where the vertical line indicates the end of the buffer. 
			% \begin{table}[h!]
				\begin{center}  
					\begin{tabular}{lccccc|cc}
						& $g_0$ & $g_1$ & $g_2$   &\ldots & $g_{b-1}$ & $g_b$  & \ldots \\[5pt]
						agent 1: \hspace{1em} & 1 & $\varepsilon$ & $\varepsilon^2$ &  \ldots & $\varepsilon^{b-1}$ & $\varepsilon^b$ & \ldots \\
						agent 2: \hspace{1em} & 1 & $\varepsilon$ & $\varepsilon^2$ & \ldots & $\varepsilon^{b-1}$ & $\varepsilon^b$ & \ldots 
					\end{tabular}
				\end{center}
				% \end{table}
			Before we go any further, we need the following simple observation: for $\gamma = c^{\frac{b+2}{b+1}}$,  goods $g_1, \ldots, g_{b}$ are exactly what $\texttt{Adv}_{b-1}(\gamma)$ would generate, just scaled down by $c^{4(b+2)}$. Indeed, $\texttt{Adv}_{b-1}(\gamma)$ would initially generate goods $\hat{g}_0, \ldots, \hat{g}_{b-1}$ such that, for any $i\in \{0, 1, \ldots, b-1\}$:
			\begin{align*}
				v_1(\hat{g}_i) = v_2(\hat{g}_i) &= \gamma^{-4 ((b-1)+2) i} =  c^{\frac{b+2}{b+1}\cdot (-4 (b + 1) i)}\\
				& =  c^{- 4 (b+2) i} = c^{4(b+2)} c^{- 4 (b+2) (i+1)}  \\
				& = c^{4(b+2)} \cdot v_1(g_{i+1}) = c^{4(b+2)} \cdot v_2(g_{i+1}) \,.
			\end{align*}
			Also, notice that $\gamma \ge  c\ge 3$. 
			
			Given this observation, we can partially define the adversary $\texttt{Adv}_b(c)$ in terms of the adversary $\texttt{Adv}_{b-1}(\gamma)$: as long as item $g_0$ is not allocated,  $\texttt{Adv}_b(c)$ simulates $\texttt{Adv}_{b-1}(\gamma)$ but with all values scaled down by $c^{4(b+2)}$ and pretending that $g_1, \ldots, g_{b}$ are $\hat{g}_0, \ldots, \hat{g}_{b-1}$. By the inductive hypothesis, in the first $5(b-1)^2 + 4 \le 5b^2 +3$ time steps either item $g_0$ is allocated or $\mathcal{A}$  fails to maintain a $1/(\gamma-1)$-\efo (and, thus, a $1/(c-1)$-\efo) allocation
            by using at most $b+1 + 6(b-1)^2 + 3 < 6b^2 + 3$ distinct values in total (all of which are at most $1$ at this point; see below). We assume the former, i.e., $g_0$ gets allocated during some time step $t_0 \le 5(b-1)^2 + 4$. Without loss of generality, we may also assume that it was agent 1 who got $g_0$, the other case being completely symmetric.

			So, at the end of time step $t_0$, agent 1 has received an item of value $1$ and---possibly---some items of smaller value. Recall that the value of each such item is a value that $\texttt{Adv}_{b-1}(\gamma)$ could generate but scaled down by $c^{4(b+2)}$. By the inductive hypothesis, any value generated by $\texttt{Adv}_{b-1}(\gamma)$ is at most $\gamma^{(b-1)+1} = \gamma^b$. Overall,
			\begin{align*}
				v_1(A^{t_0}_1) &\le 1 + (5(b-1)^2 + 4)\cdot \gamma^b \cdot c^{-4(b+2)} %\\
				= 1 + (5(b-1)^2 + 4)\, c^{-\frac{(3b+4)(b+2)}{b+1}} \\
				& < 1 + (5(b-1)^2 + 4)\, c^{-(3b+4)} %\\
				< 1 + 0.5\, c^{-(2b+3)} \,,
			\end{align*}	
			where the last inequality is a matter of simple calculus, as it reduces to showing that $0.5\, c^{x+1}> 5x^2-10x+9$ for all $x\ge 1$ and all $c\ge 3$.
			Similarly, at the end of $t_0$, agent 2 has received items of total value $v_2(A^{t_0}_2) \le 0.5\, c^{-(2b+3)}$. Further, each item already in the buffer (including the $t_0$-th good $g_{t_0}$) has value at most $\gamma^b \, c^{-4(b+2)}$ for either agent and, thus, the total value in the buffer, say $v_1(B_{t_{0}})$ and $v_2(B_{t_{0}})$ respectively, is at most $b \, \gamma^b \, c^{-4(b+2)} < b\, c^{-(3b+4)} < 0.5\, c^{-(2b+4)}$, where the last inequality follows from the simple fact that $0.5\, c^{x}> x$ for all $x\ge 1$ and all $c\ge 3$. 
			
			From this point onward,  $\texttt{Adv}_b(c)$ generates items with values that are `large' compared to $v_1(A^{t_0}_1)$ and $v_2(A^{t_0}_2)$. The next $b+2$ items are 
			$g_{t_0+1}, g_{t_0+2}, \ldots, g_{t_0+ b+2}$ but we rename them to $h_1, h_2, \ldots, h_{b+2}$ to simplify the notation. These items are such that $v_1(h_i)= c^i$ for $i\in[b+1]$ and $v_1(h_{b+2})= c^{b+1}$, whereas $v_2(h_i)= c^{-b-2+i}$ for $i\in[b+2]$
			as shown below: 
			% \begin{table}[h!]
				\begin{center}  
					\begin{tabular}{lccccccc}
						& $h_1$ & $h_2$ & $h_3$   &\ldots & $h_{b}$ & $h_{b+1}$  & $h_{b+2}$ \\[5pt]
						agent 1: \hspace{1em} & $c$ & $c^2$ & $c^3$ &  \ldots & $c^{b}$ & $c^{b+1}$ & $c^{b+1}$ \\
						agent 2: \hspace{1em} & $c^{-b-1}$ & $c^{-b}$ & $c^{-b+1}$ & \ldots & $c^{-2}$ & $c^{-1}$ & $1$ 
					\end{tabular}
				\end{center}
				% \end{table}
			After this point, if needed, $\texttt{Adv}_b(c)$ generates copies of $h_{b+2}$ for $h_t$, $t> b+2$. 
			
			Eventually, by the end of time step $t_0+b+1$ algorithm $\mathcal{A}$ will be forced to allocate at least one of $h_1, h_2, \ldots, h_{b+1}$. Let $j$ be the index of the very first such item that gets allocated, say at time $t_1 \in \{t_0+1,\ldots, t_0+b+1\}$. We claim that no item among $h_{j+1}, h_{j+2}, \ldots$ can be allocated while maintaining $1/(c-1)$-\efo or, equivalently, the \textit{first} item among $h_{j+1}, h_{j+2}, \ldots$ that is allocated, forces $\mathcal{A}$ to violate $1/(c-1)$-\efo. Towards proving this claim, first notice that at the beginning of time step $t_1$ agent 2 is envious and her total value is less than $v_2(h_j)$. Indeed,  
			%\begin{align*}
			\[	v_2(A^{t_1 - 1}_2) \le v_2(A^{t_0}_2) + v_2(B_{t_{0}}) 
			< 0.5\, c^{-(2b+3)} + 0.5\, c^{-(2b+4)} 
			< c^{-(2b+3)}\,,\]
			%\end{align*}
			whereas $v_2(A^{t_1 - 1}_1) \ge v_2(g_0) = 1$ and $v_2(h_j) \ge v_2(h_1) = c^{-b-1}$. Thus, if $h_j$ was given to agent 1, the resulting allocation would not be $1/(c-1)$-\efo from the point of view of agent 2. We conclude that $\mathcal{A}$ allocates $h_j$ to agent 2. Next, following a similar argument, we claim that after $t_1$, and using only items up to $h_j$, \textit{both} agents are envious and prefer any item among $h_{j+1}, h_{j+2}, \ldots$ to their current bundle \textit{by a factor greater than} $c-1$.
			
			Consider the beginning of any time step $t > t_1$ and assume that none of $h_{j+1}, h_{j+2}, \ldots$ has been allocated. We begin with agent 1:
			\begin{align*}
				v_1(A^{t - 1}_1) &\le v_1(A^{t_0}_1) + v_1(B_{t_{0}}) + \sum_{i=1}^{j-1} v_1(h_i) %\\
				< 1+ 0.5\, c^{-(2b+3)} + 0.5\, c^{-(2b+4)} + \sum_{i=1}^{j-1} c^i \\
				& < 1+ c^{-(2b+3)} + \frac{c^j-c}{c-1}%\\
				= \frac{c^j}{c-1} + \Big( c^{-(2b+3)} - \frac{1}{c-1}\Big) %\\
				< \frac{c^j}{c-1} \le \frac{v_1(h_\ell)}{c-1} \,,
			\end{align*}
			for any $\ell\ge j$, also implying that $(c-1) \, v_1(A^{t - 1}_1) < v_1(h_j) \le v_1(A^{t - 1}_2)$. Next, for agent 2:
			\begin{align*}
				v_2(A^{t - 1}_2) &\le v_2(A^{t_0}_2) + v_2(B_{t_{0}}) + \sum_{i=1}^{j} v_2(h_i) %\\
				< c^{-(2b+3)} + \sum_{i=1}^{j} c^{-b-2+i} \\
				& = c^{-(2b+3)} + \frac{c^{-b-1+j}- c^{-b-1}}{c-1}%\\
				\le \frac{c^{-b-1+j}}{c-1} + \big( c^{-2b-3} - c^{-b-2}\big) \\
				&< \frac{c^{-b-1+j}}{c-1} \le \frac{v_2(h_\ell)}{c-1} \,,
			\end{align*}
			for any $\ell\ge j+1$, also implying that $(c-1) \, v_2(A^{t - 1}_2) < v_2(h_{j+1}) \le 1 \le v_2(A^{t - 1}_1)$.
			
			We conclude that no item among $h_{j+1}, h_{j+2}, \ldots$ can be allocated without violating $1/(c-1)$-\efo. However, within at most $b+1$ time steps after $t_1$, $\mathcal{A}$  will be forced to allocate such an item, failing to maintain a $1/(c-1)$-temporal-\efo allocation. Note that this happened in at most 
			\[t_0 + (t_1 - t_0)  + b + 1 \le 5(b-1)^2 + 4  +b + 1 + b + 1 = 5b^2 -10b +5 + 6 +2b\le 5b^2 + 3\]
			time steps, as claimed. Further, the largest value generated by $\texttt{Adv}_b(c)$ is $c^{b+1}$, the value that agent 1 has for items $h_{b+1}, h_{b+2}, \ldots$. Finally,  the number of distinct values used is at most the $b+1$ values of items $g_0, \ldots, g_{b}$, the $2b+2$ values of items $h_1, \ldots, h_{b+2},\ldots$ other than $1$, and all the values generated by $\texttt{Adv}_{b-1}(\gamma)$, which by the inductive hypothesis are at most $6(b-1)^2+3$, for a total of at most
			\[3b + 3 + 6(b-1)^2 + 3 \le  6b^2 -12b +6 + 6 +3b\le 6b^2 + 3\]
			as claimed. This completes the induction. The statement follows by setting $b = \lfloor\sqrt{(k - 3)/6}\rfloor$ for any given $k\ge3$, since then the adversary uses at most $6b^2+3\le k$ distinct values and $\mathcal{A}$ fails within $5b^2+3\le (5k+3)/6$ time steps, i.e., it maintains a $\beta$-temporal-\efo allocation for at most $(5k-3)/6 < 5k/6$ time steps.
		\end{proof}

%%%%%%%%%%%%%%%%%%%%%%%%%%%%%%%%%%%%%%%%%%
\section{General Additive Valuation Functions}
\label{sec:general_additive}
%%%%%%%%%%%%%%%%%%%%%%%%%%%%%%%%%%%%%%%%%%

So far, we have focused on $k$-value instances but, of course, the ultimate goal is to be able to say something meaningful about instances in which agents have additive valuation functions without any restrictions. Indeed, there is a natural way to use $k$-value instances as a proxy for this: create a discretization of each agent's range of values using $k$ appropriately selected values, round everything up or down to get a $k$-value instance, construct an allocation for the latter, and, finally translate the fairness guarantee for the discretized valuations to an \textit{approximate} guarantee for the original valuation functions. Such approaches have been used for either goods with values in an interval $[1,C]$ for $k=2$ \citep{amanatidis2025online} or chores with strictly negative values for $k = \lceil \log_2 \rho \rceil$ \citep{SongTWZ25}, where $\rho$ is the max-ratio parameter defined below.
Here we fully generalize these approaches and show how our results of Section \ref{sec:main} can be applied to general additive instances with only goods  or only chores. This restriction is imposed by the fairness notion itself: as bundle values may be positive, negative, or zero, a multiplicative relaxation of \efo like the one of Definition~\ref{def:rho-EF1} is not meaningful for mixed manna. We stress that our exact results, Theorems~\ref{thm:sas_positive} and \ref{thm:dab_positive}, do apply to mixed manna.

One first simple observation here is that zero values should be dealt with separately. If one rounds any positive values to $0$ or the other way around, any guarantee for the rounded instance may completely fail for the original.  
\citet{amanatidis2025online} and \citet{SongTWZ25} both work with strictly positive or strictly negative values, respectively,  but there is a simple fix for this: one of the values of our discretization should be $0$ (thus, capturing all the values that are exactly $0$), whereas all positive values should be rounded up and all negative values should be rounded down. 
A relevant parameter is the largest ratio between two non-zero values across all agents, assuming there are such values; we call this \textit{max ratio} for short:
% $\rho = \max_{i\in N}\max_{g,h : v_i(h)\neq 0} {v_i(g)}/{v_i(h)}$
\[\rho = \max_{i\in N}\max_{g,h : v_i(h)\neq 0} \frac{v_i(g)}{v_i(h)}\,,\]
where we use the convention that the max of the empty set is $0$ to cover the case where some agents see everything as zero-valued. Of course, dealing with instances where $\rho = 0$ is trivial, so we care for the case where $\rho > 0$.

\begin{theorem}\label{thm:reduction}
Let $k\ge 2$ be an integer. For any goods-only or chores-only instance with a max ratio $\rho>0$, there is a reduction to a $k$-value instance, so that any temporal guarantee with respect to \ef or \efo we may obtain for the latter (e.g., via Algorithm \ref{alg:online-buffer}) can be translated to the corresponding temporal guarantee for the original instance, at the cost of an additional multiplicative factor of $\rho^{-1/(k-1)}$.
\end{theorem}

\begin{proof}%[\textbf{Proof of Theorem \ref{thm:reduction}}]
We construct auxiliary valuation functions that approximate the original valuation functions with only $k$ values. First, notice that because \ef and \efo are scale-free, it is without loss of generality to assume that in a goods-only instance all values an agent $i$ has for any good belong to the set $\{0\} \cup [1,\rho]$, by dividing everything by $i$'s smallest positive value, if such a value exists. Similarly, in a chores-only instance we assume that $v_i(g) \in [-\rho, -1] \cup \{0\}$ for all $i\in N$ and all chores $g\in M$.

Now, given an additive valuation function $v_i$ of an agent $i$, such that $v_i(g)\in [-\rho, -1] \cup \{0\} \cup [1,\rho]$, for all $g\in M$, we define the $k$-value \emph{threshold}  function $\hat{v}_i$ as follows:
	\[\hat{v}_{i}(g)  = \begin{cases} 	\rho^{\frac{j}{k-1}} \,, & \text{if }v_i(g) \in (\rho^{\frac{j-1}{k-1}}, \rho^{\frac{j}{k-1}}] \text{ with } j\in \{2, 3, \ldots, k-1\} \\ 
										\rho^{\frac{1}{k-1}} \,, & \text{if }v_i(g) \in [1, \rho^{\frac{1}{k-1}}] \\ 
										0\,, & \text{if } v_i(g) = 0\\ 
										-\rho^{\frac{1}{k-1}} \,, & \text{if }v_i(g) \in [-\rho^{\frac{1}{k-1}}, -1] \\
										-\rho^{\frac{j}{k-1}} \,, & \text{if }v_i(g) \in [-\rho^{\frac{j}{k-1}}, -\rho^{\frac{j-1}{k-1}}) \text{ with } j\in \{2, 3, \ldots, k-1\}
										\end{cases}\]
	for any $g\in M$. Notice that here we allow $2k-1$ values in the definition of $\hat{v}_{i}$ to avoid repetition. However, since we assume exclusively goods-only or chores-only instances, $\hat{v}_{i}$ in such instances can only have the $k$ values of the top three or of the bottom three branches.

	Next, we claim that in a goods-only instance, for any set of items $S\subseteq M$ and any agent $i\in N$, it holds that $\rho^{\frac{-1}{k-1}} \hat{v}_{i}(S)  \le v_i(S) \le \hat{v}_{i}(S)$. To see this, first note that $v_i(g)$ is rounded up in order to obtain $\hat{v}_{i}(g)$, so  $v_i(g) \le \hat{v}_{i}(g)$.  
	Moreover, the rounding factor satisfies $\hat{v}_i(g)/v_i(g) \le \rho^{\frac{1}{k-1}}$ (by inspection of the threshold function), so $\rho^{\frac{-1}{k-1}} \hat{v}_i(g) \le  v_i(g)$.
	These inequalities extend to any set of goods, as both ${v}_{i}$ and $\hat{v}_{i}$ are additive.

	Similarly, in a chores-only instance, for any set of items $S\subseteq M$ and any agent $i\in N$, it holds that $ \hat{v}_{i}(S)  \le v_i(S) \le \rho^{\frac{-1}{k-1}} \hat{v}_{i}(S)$.

	Now, suppose that at the end of some time step $t$
	the allocation $(A_1^t, \ldots, A_n^t)$ is $\alpha$-\ef or $\alpha$-\efo with respect to the threshold functions $\hat{v}_1, \ldots, \hat{v}_n$. 
	
	We first argue for the case of goods-only instances, where the parameter $\lambda$ below is $0$ in the case of $\alpha$-\ef and $1$ in the case of $\alpha$-\efo.
	For any $i, j \in N$, we have
	\[v_i(A^{t}_i)   \ge \rho^{\frac{-1}{k-1}}\, \hat{v}_i(A^{t}_i) \ge  \rho^{\frac{-1}{k-1}} \alpha  \min_{S:|S|\le \lambda} \hat{v}_i(A^{t}_j \setminus S)
	\ge \rho^{\frac{-1}{k-1}} \alpha\, \min_{S:|S|\le \lambda} v_i(A^{t}_j \setminus S)\,. \]

	Similarly, for the case of chores-only instances (where again $\lambda$  is $0$  for \ef and $1$ for \efo), for any $i, j \in N$, we have
	\[\rho^{\frac{-1}{k-1}} \alpha  \max_{S:|S|\le \lambda} v_i(A^{t}_i \setminus S)   \ge \rho^{\frac{-1}{k-1}} \alpha  \max_{S:|S|\le \lambda} \hat{v}_i(A^{t}_i \setminus S)  \ge \rho^{\frac{-1}{k-1}}\, \hat{v}_i(A^{t}_j) \ge  v_i(A^{t}_j)\,. \]
	
	Thus, in both cases, $(A_1^t, \ldots, A_n^t)$ is $\rho^{\frac{-1}{k-1}} \alpha$-\ef or $\rho^{\frac{-1}{k-1}} \alpha$-\efo, respectively,  with respect to the original functions ${v}_1, \ldots, {v}_n$.
\end{proof}

Given any goods-only or chores-only additive instance, we call the instance constructed in the proof of Theorem \ref{thm:reduction} the \textit{$k$-auxiliary instance}. Note that running Algorithm~\ref{alg:online-buffer} on the $k$-auxiliary instance requires knowing $\rho$ and each agent's smallest non-zero absolute value in advance, so that the rounded values can be computed as the items arrive.
Combining the theorem with Theorems \ref{thm:sas_positive} and \ref{thm:dab_positive}, we directly get the following corollaries.

\begin{corollary}\label{cor:buffer_additive}
For any goods-only or chores-only instance with a max ratio $\rho>0$, Algorithm \ref{alg:online-buffer} with a buffer of size $(n-1)k$ on the $k$-auxiliary instance computes a  $\rho^{-1/(k-1)}$-temporal-\efo allocation in the SAS model, and a $\rho^{-1/(k-1)}$-temporal-\ef/\efo allocation in the DAB model,
with respect to the original valuation functions.
\end{corollary}

\begin{corollary}\label{cor:buffer_additive_asymptotic}
For any goods-only or chores-only instance with a max ratio $\rho>0$, Algorithm \ref{alg:online-buffer} with a buffer of size $\Theta(n\log \rho)$ on the corresponding $\Theta(\log \rho)$-auxiliary instance computes  a $\Omega(1)$-temporal-\efo allocation in the SAS model, and a $\Omega(1)$-temporal-\ef/\efo allocation in the DAB model, 
with respect to the original valuation functions.
\end{corollary}
More specifically, for $k=\max\{2,\lceil\log_2\rho\rceil+1\}$ we have $\rho^{-1/(k-1)}\ge 1/2$, so Corollary~\ref{cor:buffer_additive} gives $1/2$-temporal-\efo and $1/2$-temporal-\ef/\efo allocations, respectively.

In order to put Corollaries \ref{cor:buffer_additive} and \ref{cor:buffer_additive_asymptotic} into perspective, 
recall that any negative result about $k$-value instances directly transfers to general additive instances as well. Specifically, Theorems \ref{thm:3-value_impossibility}, \ref{thm:buffer-impossibility-batches} and \ref{thm:buffer-impossibility-single_item} imply the following analogs for additive instances.

\begin{corollary}\label{cor:3-value_impossibility_additive}
Let $\varepsilon>0$. There is no deterministic online algorithm without a buffer that can always compute $(1/\sqrt{\rho} + \varepsilon)$-temporal-\efo allocations for additive instances with a max ratio $\rho>0$, even  when $n=2$ and all items are only goods or only chores.   
\end{corollary}

\begin{corollary}\label{cor:buffer-impossibility-batches_additive}
Let $\mathcal{A}$ be a deterministic online algorithm in the DAB model that uses a buffer of  size up to $m-1$. Then $\mathcal{A}$ may fail to produce an \ef allocation in at least half of the time steps it updates the allocation for additive instances, even if all items are goods or chores.  
\end{corollary}

\begin{corollary}\label{cor:buffer-impossibility-single_item_additive}
Let $\beta \in (0,1]$ and $\mathcal{A}$ be a deterministic online algorithm in the SAS model that uses a buffer of size up to $\lfloor\sqrt{(m - 3)/6}\rfloor$. 
Then $\mathcal{A}$ cannot always maintain a $\beta$-temporal-\efo allocation for additive instances, even when $n=2$ and all items are goods or chores. 
\end{corollary}

Corollaries \ref{cor:buffer-impossibility-batches_additive} and \ref{cor:buffer-impossibility-single_item_additive}, in particular, suggest that large buffers are necessary in order to obtain strong guarantees, even for just two agents.

%%%%%%%%%%%%%%%%%%%%%%%%%%%%%%%%%%%%%%%%%
\section{Discussion and Open Questions}
\label{sec:discussion}
%%%%%%%%%%%%%%%%%%%%%%%%%%%%%%%%%%%%%%%%%
In this work, we introduced and systematically studied the problem of fair division of indivisible mixed manna among agents with additive valuations in an online setting where algorithms are equipped with buffers that can store and rearrange items.
We placed particular emphasis on $k$-value instances and showed that buffers of size linear in $k$ and in the number of agents suffice to obtain strong guarantees via novel combinatorial arguments. These results extend to general additive goods-only or chores-only instances, at the cost of some instance-dependent loss. In contrast to much of the existing literature---which circumvents strong impossibility results by severely restricting the space of instances---we instead enhance the power of online algorithms. Despite the generality of our approach, it opens up several interesting directions for future work.

A natural direction is to aim for positive results with smaller buffers by relaxing the fairness requirements, for example by targeting approximate temporal-\efo in the DAB model or approximate \efo only \textit{once every few steps} in the SAS model. Our impossibility results in Section~\ref{sub:impossibility_buffer} do not rule out such guarantees. Even for temporal-\efo in the SAS model, there is a gap between the buffer of size $(n-1)k$ used in Theorem~\ref{thm:sas_positive} and the lower bound of Theorem~\ref{thm:buffer-impossibility-single_item}, which only shows that a buffer of size $\Omega(\sqrt{k})$ is necessary, even for two agents; closing this gap is probably the most challenging question left open by our work.
Another idea for obtaining strong results while reducing buffers is to combine them with lookahead. That is, allow online algorithms to look $\ell$ steps into the future, while maintaining a  buffer of size $b$, potentially significantly smaller than $\ell$. Again, this model is not captured by our impossibility results in Sections \ref{sec:impossibility_lookahead} or \ref{sub:impossibility_buffer}. In a different direction, one could also require that each item is allocated within a bounded number of time steps after its arrival; our buffers bound the number of items that are delayed at any time, but not the delay of each individual item.
Finally, while our work focuses on envy-based  fairness notions, share-based notions---such as maximin share (MMS) fairness---are equally relevant. It would be interesting to investigate whether buffer-augmented online algorithms can achieve approximate temporal-MMS guarantees using buffers of comparable, or even smaller, size.

\subsection*{Acknowledgments}
    This work has been partially supported by the project MIS 5154714 of the National Recovery and Resilience Plan ``Greece 2.0'' funded by the European Union under the NextGenerationEU Program. It was also supported by Input Output (iog.io) through their funding of the Athena Blockchain Technology Lab. The authors are also grateful to Georgios Korpas, Jeremy Kulcsar and William Shoosmith from the Quantum and Advanced Compute Group at HSBC, and to Dimitris Fotakis, Marina Kontalexi, and Alexandros Lolos from Archimedes, Athena Research Center for stimulating discussions and feedback on the work.

\subsection*{Disclaimer}   
    This paper was prepared for information purposes and is not a product of HSBC Bank Plc.~or its affiliates. Neither HSBC Bank Plc.~nor any of its affiliates make any explicit or implied representation or warranty and none of them accept
    any liability in connection with this paper, including, but not limited to, the completeness, accuracy, reliability of information contained herein and the potential legal, compliance, tax
    or accounting effects thereof. Copyright HSBC Group 2026.

%%%%%%%%%%%%%%%%%%%%%%%%%%%%%%%%%%%%
\bibliography{bibliography}
%%%%%%%%%%%%%%%%%%%%%%%%%%%%%%%%%%%%

%%%%%%%%%%%%%%%%%%%%%%%%%%%%%%%%%%%%
\newpage

\appendix

%%%%%%%%%%%%%%%%%%%%%%%%%%%%%%%%%%%
\section{The Double Round-Robin Algorithm}
\label{sec:DRRisTEF1}
%%%%%%%%%%%%%%%%%%%%%%%%%%%%%%%%%%%

To  complete the correctness of  Algorithm \ref{alg:online-buffer} in the SAS model we need to argue that running the Double Round-Robin algorithm of~\citet{aziz2022fair} on the items left in the buffer can be turned into a temporal-EF1 allocation.

Let $B$ be the set of left-over items on which we apply Double Round-Robin.
The algorithm is based on two picking sequences. First, the items that give non-positive utility to every agent are allocated according to a fixed round-robin order, say $1, 2, \ldots,  n$. We call this set of items $B^- = \{g \in B : \forall i \in N,\ v_i(g) \le 0\}$. Dummy null items may be added to this set so that its size is a multiple of $n$. Then, the remaining items, namely those that give strictly positive utility to at least one agent, i.e., $B^+ = \{g \in B \mid \exists i \in N \text{ such that } v_i(g) > 0\}$,
are allocated according to the reverse round-robin order, $n, n-1, \ldots, 1$. In this second phase, if the current agent has no available item that gives her strictly positive utility, she pretends to pick a dummy item of value $0$ instead. Finally, all dummy items are removed, and the resulting allocation is returned. A pseudocode description of the algorithm is given in Algorithm 1 in~\cite{aziz2022fair}.

What is shown by \citet{aziz2022fair} is that its final allocation is EF1 for mixed manna.

\begin{theorem}[\citet{aziz2022fair}]\label{thm:aziz_DDR_is_EF1}
The Double Round-Robin algorithm returns an EF1 allocation.
\end{theorem}

Theorem \ref{thm:aziz_DDR_is_EF1} is enough for completing the proof of Theorem \ref{thm:dab_positive} but for Theorem \ref{thm:sas_positive} we need something stronger; namely that the output of Double Round-Robin can be turned into a temporal-EF1 allocation.
However, since Double Round-Robin allocates the items sequentially and we have all items in $B$ available offline, it suffices to show that during its execution Double Round-Robin maintains an \efo allocation in every step. 
In the following lemma we  show exactly this.

\begin{lemma}\label{lem:DRR_is_TEF1}
    Double Round-Robin run on a set $S$ of items builds its output allocation one item at a time and every intermediate partial allocation is EF1.
\end{lemma}

\begin{proof}
We use Theorem \ref{thm:aziz_DDR_is_EF1} as a black box. 
Consider the set $S$ of items and let $z=|S|$.
Let
\[
    \pi=(g_1,\ldots,g_z)
\]
be the sequence of real items (i.e., items of $S$) allocated by the algorithm, in the order in which
they are selected. For every $q \leq z$, let $\mathcal{A}^q=(A^q_1,...,A^q_n)$ be the allocation after the
first $q$ real items of $\pi$ have been allocated.

We show that $\mathcal{A}^q$ is EF1 for all $q \in \{1,...,z\}$. Let
% \[
    $S^q=\{g_1,\ldots,g_q\}$
% \]
be the set of real items allocated up to this point. Let $S^- = \{g \in S : \forall i \in N,\ v_i(g) \le 0\}$ and $S^+ = S\setminus S^-$ be the sets of items allocated in the first and in the second phase, respectively (i.e., the sets $B^-$ and $B^+$ defined above, for $B=S$). We distinguish two cases.

If $q\le |S^-|$, then $\aaa^q$ consists of the real items picked in a prefix of the first phase, i.e., in a prefix of a round-robin sequence over items (real or dummy) that every agent values non-positively. Fix two agents $i,j$, and let $a_1,a_2,\ldots$ and $b_1,b_2,\ldots$ be the items picked by $i$ and $j$, respectively, in this prefix, in the order they were picked. If $i$ precedes $j$ in the order $1,2,\ldots,n$, then $a_t$ is picked before $b_t$, so $v_i(a_t)\ge v_i(b_t)$ for every $t$ for which $b_t$ exists, and $i$ has picked as many items as $j$ or one more; hence $v_i(A^q_i)\ge v_i(A^q_j)$, possibly after removing the last item of $i$ from $A^q_i$. If $j$ precedes $i$, then $a_t$ is picked before $b_{t+1}$, so $v_i(a_t)\ge v_i(b_{t+1})$ for every $t$ for which $b_{t+1}$ exists, and $j$ has picked as many items as $i$ or one more; since $v_i(b_1)\le 0$, we again get $v_i(A^q_i)\ge v_i(A^q_j)$, possibly after removing the last item of $i$ from $A^q_i$. As dummy items have value $0$ for everyone, removing a dummy item is never needed, and thus $\aaa^q$ is EF1.

If $q> |S^-|$, then $S^q$ contains all the items of $S^-$. The main observation is
that $\aaa^q$ is exactly the allocation produced by Double Round-Robin on the
restricted instance with item set $S^q$, with the same tie-breaking as in the
original execution.

Indeed, this restricted instance is obtained from the original one by deleting
the suffix $g_{q+1},\ldots,g_z$ of the allocation sequence. This does not change
any of the choices made before $g_q$ is allocated. Whenever an agent selects an
item, she selects a most preferred available item from the relevant set of
remaining items, or a dummy item. Hence, after deleting only items that would be selected later,
the item selected at each earlier step is still a valid choice. 
Moreover, the partition of the items into sets $S^+$ and $S^-$ is item specific. Therefore, deleting other items does not
change whether a remaining item belongs to the first or the second phase of the
algorithm. In particular, since $S^q$ contains all the items of $S^-$, exactly as many dummy items are added in the first phase as in the original execution. Dummy null items have value $0$ for every agent and are removed at
the end, so they do not affect the allocation of real items.

Thus, the partial allocation $\mathcal{A}^q$ is the same as the final allocation that
Double Round-Robin would return when run on the set of items $S^q$. By
Theorem \ref{thm:aziz_DDR_is_EF1}, this allocation is EF1. 
Since this holds for any arbitrary $q$, every prefix of the allocation sequence is EF1, and the lemma follows.
\end{proof}

%%%%%%%%%%%%%%%%%%%%%%%%%%%%%%%%%%%%

\end{document}